\documentclass[reprint,NumberedRefs]{JASA}
\def\theDOI{10.1121/10.0017252} 
\editorinitials{Content identical to JASA article \dodoi{10.1121/10.0017252}}

\usepackage[utf8]{inputenc}
\usepackage{lmodern}
\usepackage[]{graphicx}
\usepackage{amsmath}
\usepackage{amssymb}
\usepackage{amsthm}
\usepackage{array} 
\usepackage{url} 
\usepackage[USenglish]{babel}
\usepackage{siunitx}
\DeclareSIUnit{\permille}{\text{\textperthousand}}
\usepackage[]{hyperref} 

\RequirePackage{tikz,pgfplots,grffile,pgfplotstable} 
\pgfplotsset{compat=1.17} 
\SendSettingsToPgf 
\usepgfplotslibrary{external}
\tikzexternalenable
\tikzset{external/system call={lualatex \tikzexternalcheckshellescape -halt-on-error -interaction=batchmode -jobname "\image" "\texsource"}}
\usetikzlibrary{positioning}

\pgfplotsset{/pgfplots/colormap={comsolthermal}{[1cm] rgb255(0cm)=(255,255,255) rgb255(3cm)=(255,255,0) rgb255(6cm)=(255,0,0) rgb255(8cm)=(100,0,0)}}
\pgfplotsset{/pgfplots/colormap={reverse viridis}{indices of colormap={\pgfplotscolormaplastindexof{viridis},...,0 of viridis}}}
\pgfplotsset{colormap/viridis} 
\pgfplotsset{
colormap={inferno}{[1pt] 
rgb(0pt)=(0.00146159,0.000466128,0.0138655); rgb(1pt)=(0.00399889,0.00287953,0.0280795); rgb(2pt)=(0.00795128,0.00636563,0.0480947); rgb(3pt)=(0.0133413,0.0107186,0.0695703); rgb(4pt)=(0.0202832,0.0157475,0.0913609); rgb(5pt)=(0.0289796,0.021235,0.113568); rgb(6pt)=(0.0397217,0.0269215,0.136281); rgb(7pt)=(0.0517865,0.0325368,0.159532); rgb(8pt)=(0.0643438,0.0377499,0.18329); rgb(9pt)=(0.0775923,0.0421673,0.207524); rgb(10pt)=(0.0916217,0.0453405,0.232044); rgb(11pt)=(0.106525,0.0471492,0.256591); rgb(12pt)=(0.122332,0.047544,0.280778); rgb(13pt)=(0.138952,0.0465727,0.304032); rgb(14pt)=(0.156258,0.0445008,0.32582); rgb(15pt)=(0.173955,0.0418068,0.345479); rgb(16pt)=(0.191824,0.0391205,0.362738); rgb(17pt)=(0.209623,0.0371602,0.377434); rgb(18pt)=(0.227215,0.0364061,0.389682); rgb(19pt)=(0.244559,0.0370282,0.399785); rgb(20pt)=(0.261616,0.0391236,0.407959); rgb(21pt)=(0.278437,0.0424959,0.414594); rgb(22pt)=(0.295039,0.0468597,0.419873); rgb(23pt)=(0.311475,0.0519108,0.424059); rgb(24pt)=(0.327775,0.057427,0.427304); rgb(25pt)=(0.343975,0.0632266,0.429737); rgb(26pt)=(0.360096,0.0691767,0.431479); rgb(27pt)=(0.376175,0.0751899,0.432574); rgb(28pt)=(0.392211,0.0812091,0.433118); rgb(29pt)=(0.408237,0.0871945,0.433122); rgb(30pt)=(0.424255,0.0931259,0.432631); rgb(31pt)=(0.440265,0.0989922,0.431672); rgb(32pt)=(0.456291,0.104796,0.430239); rgb(33pt)=(0.472328,0.110547,0.428334); rgb(34pt)=(0.488378,0.116242,0.425992); rgb(35pt)=(0.504438,0.121907,0.423183); rgb(36pt)=(0.520504,0.127551,0.419927); rgb(37pt)=(0.536573,0.133198,0.416213); rgb(38pt)=(0.552636,0.138866,0.412039); rgb(39pt)=(0.568685,0.144579,0.407414); rgb(40pt)=(0.58471,0.150362,0.402323); rgb(41pt)=(0.600697,0.156242,0.396784); rgb(42pt)=(0.616635,0.162246,0.390784); rgb(43pt)=(0.632507,0.168405,0.384337); rgb(44pt)=(0.648296,0.174748,0.377445); rgb(45pt)=(0.663988,0.181304,0.370113); rgb(46pt)=(0.679556,0.188111,0.36236); rgb(47pt)=(0.694987,0.195192,0.354191); rgb(48pt)=(0.710252,0.202589,0.345625); rgb(49pt)=(0.725334,0.210323,0.336673); rgb(50pt)=(0.740206,0.21843,0.327359); rgb(51pt)=(0.754841,0.226936,0.317702); rgb(52pt)=(0.76922,0.235863,0.307723); rgb(53pt)=(0.783305,0.245244,0.297441); rgb(54pt)=(0.797086,0.255086,0.286885); rgb(55pt)=(0.81052,0.265418,0.276073); rgb(56pt)=(0.823595,0.276225,0.265041); rgb(57pt)=(0.836281,0.287556,0.253792); rgb(58pt)=(0.84855,0.299388,0.242357); rgb(59pt)=(0.860396,0.311742,0.230744); rgb(60pt)=(0.871771,0.324575,0.218982); rgb(61pt)=(0.882687,0.33792,0.207066); rgb(62pt)=(0.893103,0.351733,0.195014); rgb(63pt)=(0.903021,0.366017,0.182821); rgb(64pt)=(0.912424,0.38075,0.170486); rgb(65pt)=(0.921294,0.395903,0.158006); rgb(66pt)=(0.929644,0.411479,0.145367); rgb(67pt)=(0.937433,0.427418,0.132564); rgb(68pt)=(0.944691,0.443736,0.119586); rgb(69pt)=(0.95139,0.460386,0.106433); rgb(70pt)=(0.957538,0.47736,0.0931233); rgb(71pt)=(0.963133,0.494637,0.0797067); rgb(72pt)=(0.968159,0.512184,0.0663111); rgb(73pt)=(0.972638,0.530008,0.0531745); rgb(74pt)=(0.976535,0.54806,0.040851); rgb(75pt)=(0.979876,0.566352,0.0308821); rgb(76pt)=(0.982642,0.584853,0.0250467); rgb(77pt)=(0.984832,0.603555,0.0238049); rgb(78pt)=(0.986452,0.622437,0.0276119); rgb(79pt)=(0.987475,0.641499,0.0373252); rgb(80pt)=(0.987927,0.660706,0.0521491); rgb(81pt)=(0.987772,0.680075,0.0698596); rgb(82pt)=(0.987033,0.699564,0.0894211); rgb(83pt)=(0.985695,0.719174,0.110455); rgb(84pt)=(0.983755,0.738892,0.132774); rgb(85pt)=(0.981237,0.758673,0.1563); rgb(86pt)=(0.978132,0.778541,0.181134); rgb(87pt)=(0.974497,0.798402,0.207292); rgb(88pt)=(0.970345,0.818284,0.235014); rgb(89pt)=(0.965794,0.838073,0.264413); rgb(90pt)=(0.960969,0.857712,0.295748); rgb(91pt)=(0.95611,0.877104,0.329299)}
}
\pgfplotsset{
colormap={wave}{[1pt] 
rgb(0pt)=(0.0196078,0.188235,0.380392); rgb(1pt)=(0.0424814,0.213009,0.426102); rgb(2pt)=(0.0622178,0.236955,0.467895); rgb(3pt)=(0.077739,0.260086,0.505806); rgb(4pt)=(0.0901576,0.282422,0.539935); rgb(5pt)=(0.100125,0.303988,0.570432); rgb(6pt)=(0.10814,0.32482,0.597485); rgb(7pt)=(0.114662,0.34496,0.621312); rgb(8pt)=(0.120144,0.364464,0.642159); rgb(9pt)=(0.12506,0.383394,0.660287); rgb(10pt)=(0.129908,0.401823,0.675977); rgb(11pt)=(0.135197,0.419832,0.68952); rgb(12pt)=(0.141437,0.437513,0.701217); rgb(13pt)=(0.149104,0.454964,0.711377); rgb(14pt)=(0.158618,0.472291,0.720316); rgb(15pt)=(0.170318,0.489608,0.728353); rgb(16pt)=(0.184453,0.507033,0.735813); rgb(17pt)=(0.201192,0.52469,0.743022); rgb(18pt)=(0.220645,0.542708,0.750312); rgb(19pt)=(0.242891,0.561218,0.758017); rgb(20pt)=(0.267998,0.580354,0.766473); rgb(21pt)=(0.295783,0.60016,0.775871); rgb(22pt)=(0.325563,0.620483,0.786068); rgb(23pt)=(0.356746,0.641131,0.796874); rgb(24pt)=(0.388862,0.661907,0.808098); rgb(25pt)=(0.421521,0.682609,0.819549); rgb(26pt)=(0.454378,0.703027,0.831034); rgb(27pt)=(0.487118,0.722949,0.842357); rgb(28pt)=(0.519435,0.742156,0.853324); rgb(29pt)=(0.551027,0.760431,0.863736); rgb(30pt)=(0.581594,0.777552,0.873396); rgb(31pt)=(0.610959,0.793422,0.882219); rgb(32pt)=(0.639151,0.808154,0.89031); rgb(33pt)=(0.666217,0.821888,0.897796); rgb(34pt)=(0.6922,0.834765,0.904803); rgb(35pt)=(0.717141,0.846932,0.911458); rgb(36pt)=(0.741082,0.858537,0.91789); rgb(37pt)=(0.764062,0.869733,0.924226); rgb(38pt)=(0.786122,0.880673,0.930595); rgb(39pt)=(0.807304,0.891515,0.937126); rgb(40pt)=(0.827646,0.902413,0.943944); rgb(41pt)=(0.847152,0.913342,0.95097); rgb(42pt)=(0.86577,0.924043,0.957864); rgb(43pt)=(0.883432,0.934237,0.964264); rgb(44pt)=(0.900059,0.943644,0.969805); rgb(45pt)=(0.915562,0.951981,0.974124); rgb(46pt)=(0.929844,0.958965,0.976854); rgb(47pt)=(0.9428,0.964308,0.977629); rgb(48pt)=(0.954317,0.967728,0.976085); rgb(49pt)=(0.964274,0.968939,0.971858); rgb(50pt)=(0.972548,0.967668,0.964605); rgb(51pt)=(0.979157,0.963873,0.954273); rgb(52pt)=(0.984266,0.957739,0.941103); rgb(53pt)=(0.988048,0.949461,0.92535); rgb(54pt)=(0.990675,0.939239,0.907272); rgb(55pt)=(0.992316,0.92727,0.887125); rgb(56pt)=(0.993144,0.913752,0.865169); rgb(57pt)=(0.993327,0.89888,0.841657); rgb(58pt)=(0.993036,0.88285,0.816843); rgb(59pt)=(0.992441,0.865852,0.790977); rgb(60pt)=(0.991702,0.848066,0.764304); rgb(61pt)=(0.990786,0.829567,0.737026); rgb(62pt)=(0.989523,0.810346,0.709314); rgb(63pt)=(0.987754,0.790391,0.68134); rgb(64pt)=(0.985335,0.769692,0.653273); rgb(65pt)=(0.982134,0.748237,0.625281); rgb(66pt)=(0.978035,0.726013,0.597534); rgb(67pt)=(0.97293,0.703009,0.5702); rgb(68pt)=(0.966722,0.679209,0.543445); rgb(69pt)=(0.959326,0.654601,0.517436); rgb(70pt)=(0.950676,0.629178,0.492329); rgb(71pt)=(0.940867,0.603023,0.468165); rgb(72pt)=(0.930083,0.57627,0.444922); rgb(73pt)=(0.918502,0.549048,0.422571); rgb(74pt)=(0.906296,0.521478,0.401083); rgb(75pt)=(0.893637,0.493674,0.380426); rgb(76pt)=(0.880689,0.465738,0.360567); rgb(77pt)=(0.867618,0.437764,0.34147); rgb(78pt)=(0.854583,0.409835,0.323099); rgb(79pt)=(0.841746,0.382016,0.305418); rgb(80pt)=(0.82923,0.354351,0.288394); rgb(81pt)=(0.816912,0.3268,0.272046); rgb(82pt)=(0.804566,0.299286,0.256409); rgb(83pt)=(0.791973,0.271728,0.241519); rgb(84pt)=(0.778924,0.244033,0.22741); rgb(85pt)=(0.765216,0.21609,0.214113); rgb(86pt)=(0.750653,0.187753,0.201656); rgb(87pt)=(0.735043,0.158809,0.190065); rgb(88pt)=(0.718202,0.128912,0.17936); rgb(89pt)=(0.699949,0.0974024,0.169557); rgb(90pt)=(0.680108,0.0627338,0.16067); rgb(91pt)=(0.658509,0.0222288,0.152704); rgb(92pt)=(0.634987,0,0.145661); rgb(93pt)=(0.609383,0,0.139538); rgb(94pt)=(0.581545,0,0.134326); rgb(95pt)=(0.551327,0,0.130015); rgb(96pt)=(0.518589,0,0.126594); rgb(97pt)=(0.483198,0,0.124051); rgb(98pt)=(0.445022,0,0.122378); rgb(99pt)=(0.403922,1.01308e-16,0.121569)}
}
\pgfplotsset{
colormap={reverse inferno}{
    indices of colormap={
        \pgfplotscolormaplastindexof{inferno},...,0 of inferno}
},
}

\definecolor{inferno0}{rgb}{0, 0, 0}
\definecolor{inferno1}{rgb}{0.001461590960000,   0.000466127766000,   0.013865520000000}
\definecolor{inferno2}{rgb}{0.308956285665458,   0.051096710436248,   0.423484299000000}
\definecolor{inferno3}{rgb}{0.674779572538462,   0.185988386967698,   0.364790491307586}
\definecolor{inferno4}{rgb}{0.948363139076923,   0.452656844940405,   0.112525223996407}
\definecolor{inferno5}{rgb}{0.954311759692308,   0.884450513620312,   0.342825456922373}
\definecolor{viridis0}{rgb}{0, 0, 0}
\definecolor{viridis1}{rgb}{0.267004010000000,   0.004874330000000,   0.329415190000000}
\definecolor{viridis2}{rgb}{0.230223510512892,   0.321288876115496,   0.545488357500000}
\definecolor{viridis3}{rgb}{0.128151799694423,   0.565106970000000,   0.550904139912083}
\definecolor{viridis4}{rgb}{0.362853991628797,   0.786695125000000,   0.386551780990772}
\definecolor{viridis5}{rgb}{0.993247890000000,   0.906156570000000,   0.143936200000000}
\definecolor{blue_dark}{rgb}{0.121, 0.466, 0.706}
\definecolor{blue_light}{rgb}{0.682, 0.780, 0.910}
\definecolor{orange_dark}{rgb}{1.000, 0.498, 0.055}
\definecolor{orange_light}{rgb}{1.000, 0.733, 0.476}
\definecolor{green_dark}{rgb}{0.172, 0.627, 0.173}
\definecolor{green_light}{rgb}{0.596, 0.874, 0.541}
\definecolor{red_dark}{rgb}{0.839, 0.153, 0.157}
\definecolor{red_light}{rgb}{1.000, 0.596, 0.588}
\definecolor{violet_dark}{rgb}{0.580, 0.404, 0.741}
\definecolor{violet_light}{rgb}{0.772, 0.690, 0.835}
\definecolor{matlabBlue}{rgb}{0.00000,0.44700,0.74100}
\definecolor{matlabOrange}{rgb}{0.85000,0.32500,0.09800}
\definecolor{matlabYellow}{rgb}{0.92900,0.69400,0.12500}
\definecolor{matlabGreen}{rgb}{0.46600,0.67400,0.18800}
\colorlet{myGray}{white!40!black}
\definecolor{myLightGreen}{HTML}{dae8b5}

\colorlet{color_trapped}{viridis3}
\colorlet{color_leaky}{inferno2}
\colorlet{color_leaky_forward}{inferno2}
\colorlet{color_leaky_backward}{violet_dark}

\colorlet{color_plate}{gray!30!white}
\colorlet{color_fluid}{blue_light!50!white}
\colorlet{color_pml}{inferno5!50!white}
\colorlet{color_glue}{violet_light!70!white}

\colorlet{dark}{blue_dark}
\colorlet{light}{orange_dark}
\colorlet{verylight}{green_light}

\pgfplotscreateplotcyclelist{tab10}{blue_dark, orange_dark,green_dark,red_dark,violet_dark}
\pgfplotscreateplotcyclelist{infvir6}{
dark,
light,
verylight, 
viridis4,
inferno4 
}
\pgfplotscreateplotcyclelist{tabInferno3}{[colors of colormap={300,800,500} of inferno]}
\pgfplotscreateplotcyclelist{tabViridis5}{[colors of colormap={100,900,300,700,500} of viridis]}
\pgfplotsset{cycle list name=infvir6}

\pgfplotsset{
    /pgfplots/bar cycle list/.style={/pgfplots/cycle list={
            {viridis3,fill=viridis3!30!white,mark=none,thin},
            {inferno2,fill=inferno2!30!white,mark=none,thin},
            {viridis4,fill=viridis4!30!white,mark=none,thin},
            {inferno1!90,fill=inferno1!30!white,mark=none,thin},
            {viridis5,fill=viridis5!30!white,mark=none,thin},
            {inferno3,fill=inferno3!30!white,mark=none,thin},
}, },
}

\pgfplotsset{
    /pgfplots/semitransparent bar cycle list/.style={/pgfplots/cycle list={
            {inferno2,fill=inferno2!30!white,mark=none,thin},
            {viridis4!40!black,fill=viridis4!80!yellow,mark=none,thin, fill opacity=.4},
}, },
}

\usepgfplotslibrary{polar} 
\usetikzlibrary{shapes.geometric}
\usetikzlibrary {decorations.pathreplacing} 

\usetikzlibrary{calc}
\usetikzlibrary{intersections}
\usetikzlibrary{angles,quotes} 
\usetikzlibrary{arrows.meta}
\usetikzlibrary{matrix} 
\usetikzlibrary{shapes.misc} 
\usetikzlibrary{arrows,decorations.markings}
\tikzset{schwarzweisspunkt/.style={arrows={-Circle[width=3pt,length=3pt]}, line width=.8pt, line cap=round, preaction={draw=white,shorten >=-1pt,line width=1.5pt, opacity=.8, line cap=round, arrows={-Circle[width=5pt,length=5pt]}}}}

\tikzset{schwarzweisspfeil/.style={ arrows={-Latex[width=3pt,length=3pt]}, line width=.8pt, line cap=round, preaction={draw=white,shorten >=-1.5pt,line width=1.5pt, opacity=.8, fill opacity=.5, line cap=round, arrows={-Latex[width=5pt,length=5pt]}} }}

\tikzset{beschriftung/.style={inner sep=1pt, outer sep=2pt, rounded corners=2pt, fill=white, opacity=.5, text opacity=1}}

\tikzset{cross/.style={cross out, draw=black, inner sep=0pt, outer sep=0pt}}

\tikzset{filled arrow/.style={single arrow, fill=red!50, anchor=base, align=center,text width=1cm}}

\usetikzlibrary{shapes}
\usetikzlibrary{pgfplots.groupplots}

\pgfplotsset{
  myPlotStyle/.style={                
    every axis/.append style={        
      axis on top=false,              
      width=6cm, height=4.5cm,     
      scale only axis=true,          
      xmajorgrids, ymajorgrids,       
      legend style={                  
        font={\footnotesize\color{white!15!black}},
        legend pos=north east,
        legend cell align=left,,
        },
      cycle list name=tab10,
      mark size=1pt,                  
      line join=bevel,                
      colorbar style={
        at={(1.05, 0)}, anchor=south east,
        ylabel style={font=tiny,font=\color{white!15!black}}, 
        width=.1cm,
        },  
      },
    every axis x label/.append style={font=\color{white!15!black}},
    every axis y label/.append style={font=\color{white!15!black}},
    every axis title/.append style={font=\footnotesize}, 
    every axis plot/.append style={very thick}, 
    legend style={draw=white!15!black},
    axis background/.append style={fill=white},
    },
    pp_dashed_bw/.style={               
      postaction={draw,black,dashed},color=white,solid
    },
    pp_dashed_bw_alt/.style={           
      postaction={draw,white,dashed,dash phase=3pt},color=black,dashed
    },
}
\pgfplotsset{myPlotStyle}

\DeclareOldFontCommand{\bf}{\normalfont\bfseries}{\mathbf}
\DeclareOldFontCommand{\rm}{\normalfont\rmseries}{\mathrm}
\DeclareOldFontCommand{\tt}{\normalfont\ttseries}{\mathtt}

\DeclareMathAlphabet{\mathup}{OT1}{\familydefault}{m}{n}
\newcommand*{\mup}[1]{\mathup{#1}} 

\newcommand*{\iu}{\mathup{i}\mkern1mu}
\DeclareRobustCommand{\Re}{\operatorname{Re}}
\DeclareRobustCommand{\Im}{\operatorname{Im}}

\DeclareMathOperator{\e}{e}

\DeclareMathOperator{\rank}{rank}
\DeclareMathOperator{\nrank}{nrank}

\newcommand\ten[1]{\mathbf{#1}}

\newcommand\define{\stackrel{\text{def\ \,}}{=\ }}

\newcommand\transp{^{\top}} 

\newcommand\dirvec[1]{\mathbf{e}_{#1}}

\newtheorem{theorem}{Theorem}[section]

\begin{document}

\title{Computing zero-group-velocity points in anisotropic elastic waveguides: Globally and locally convergent methods}

\author{Daniel~A.~Kiefer}
\affiliation{Institut Langevin, ESPCI Paris, Universit\'e PSL, CNRS, 75005 Paris, France}
\email{daniel.kiefer@espci.fr}

\author{Bor~Plestenjak}
\affiliation{Faculty of Mathematics and Physics, University of Ljubljana, SI-1000 Ljubljana, Slovenia}

\author{Hauke~Gravenkamp}
\affiliation{International Centre for Numerical Methods in Engineering (CIMNE), 08034 Barcelona, Spain}

\author{Claire~Prada}
\affiliation{Institut Langevin, ESPCI Paris, Universit\'e PSL, CNRS, 75005 Paris, France}





\date{\today} 

\begin{abstract}\noindent
Dispersion curves of elastic waveguides exhibit points where the group velocity vanishes while the wavenumber remains finite. These are the so-called zero-group-velocity (ZGV) points. As the elastodynamic energy at these points remains confined close to the source, they are of practical interest for nondestructive testing and quantitative characterization of structures. These applications rely on the correct prediction of the ZGV points. In this contribution, we first model the ZGV resonances in anisotropic plates based on the appearance of an additional modal solution. The resulting governing equation is interpreted as a two-parameter eigenvalue problem. Subsequently, we present three complementary numerical procedures capable of computing ZGV points in arbitrary nondissipative elastic waveguides in the conventional sense that their axial power flux vanishes. The first method is globally convergent and guarantees to find all ZGV points but can only be used for small problems. The second procedure is a very fast, generally-applicable, Newton-type iteration that is locally convergent and requires initial guesses. The third method combines both kinds of approaches and yields a procedure that is applicable to large problems, does not require initial guesses and is likely to find all ZGV points. The algorithms are implemented in \texttt{GEW ZGV computation} (\dodoi{10.5281/zenodo.7537442}).
\end{abstract}

\maketitle

\section{Introduction} 
\label{sec:introduction}

Thin-walled mechanical structures act as elastodynamic waveguides~\cite{auld_acoustic_1990,royer_elastic_2022}. The angular frequency~$\omega$ of a guided wave is related to the wavenumber~$k$ via a dispersion relation $\omega(k)$.
There exist so-called \emph{zero-group-velocity (ZGV)} points $(\omega_*, k_*)$ on the dispersion curves where the group velocity~$c_\mup{g} = \frac{\partial \omega}{\partial k}$ vanishes while the wavenumber~$k_*$ remains finite~\cite{tassoulas_wave_1984,prada_laser-based_2005,prada_local_2008}. These are of special practical interest because the waves do not propagate, and their energy remains close to the source, leading to long-lasting ringing effects. This enables the accurate contactless assessment of structural properties such as the thickness~\cite{holland_air-coupled_2003,ces_thin_2011,baggens_systematic_2015,grunsteidl_inverse_2016}, elastic material parameters~\cite{clorennec_local_2007,ces_characterization_2012,prada_influence_2009,grunsteidl_inverse_2016,grunsteidl_determination_2018,watzl_situ_2022}, bonding state~\cite{mezil_non_2014,mezil_investigation_2015}, properties of a surrounding fluid~\cite{valenza_ii_characterizing_2021} and effective mechanical behavior, e.g., of nanoporous silicon~\cite{thelen_laser-excited_2021}. 

The use and design of the above-mentioned applications usually rely on the theoretical prediction of the ZGV points. It is common practice to extract this information from the complete theoretical dispersion curves. However, this can be a tedious task, in particular, because several digits of precision are usually required. More importantly, inverse methods for determining the sought properties rely on the automatic determination of the ZGV points~\cite{grunsteidl_inverse_2016,clorennec_local_2007,watzl_situ_2022}. For these reasons, it is desirable to develop general and efficient numerical methods to compute these points, which is the aim of the present contribution. 

Although it is not common to explicitly compute the ZGV points, several procedures have already been devised for this end. The most widespread method is based on the implicit analytical dispersion relation, which is of the form $\Omega(\omega, k) = 0$, by additionally requiring the analytically obtained group velocity $\frac{\partial \omega}{\partial k} = -\frac{\partial \Omega}{\partial k} \left(\frac{\partial \Omega}{\partial \omega}\right)^{-1}$ to vanish~\cite{grunsteidl_inverse_2016,mezil_non_2014}. Even for the simple case of an isotropic plate, this leads to relatively cumbersome expressions that require carefully implemented numerical methods to avoid numerical instabilities~\cite{royer_elastic_2022}. Conventionally, a gradient-based iterative solution method would be employed to solve this nonlinear system -- which requires one initial guess for every \emph{expected ZGV point}. The procedure has been extended to imperfectly bonded multi-layered isotropic plates~\cite{mezil_non_2014}. However, although feasible, to the best of our knowledge, it has not been employed for anisotropic plates, let alone for geometrically more complicated structures. 

Another method to obtain dispersion curves consists in performing a numerical discretization of the boundary value problem (BVP) that describes the guided waves and then solving the resulting eigenvalue problem (EVP). Based on this strategy, Kausel~\cite{kausel_number_2012} proposed to start at cut-off frequencies ($k = 0$) associated with backward waves \cite{shuvalov_backward_2008} -- which have negative group velocity for $k > 0$ -- and then follow the branch until the ZGV point is reached. To find all solutions, the method requires that the ZGV points are induced by a backward wave starting at $k = 0$ and that only one ZGV point occurs on each branch. These assumptions are not generally true. Multiple ZGV points on one branch have been observed in homogeneous but anisotropic plates~\cite{hussain_lamb_2012,karous_multiple_2019,hernando_quintanilla_modeling_2015}, layered structures~\cite{maznev_surface_2009,glushkov_multiple_2021} and fluid-filled pipes~\cite{cui_backward_2016}.
Moreover, the mentioned references have also revealed that ZGV points can lie on forward wave branches emerging at $k = 0$. We will explain this behavior in Subsec.~\ref{sub:properties_zgv_resonances}.

In this contribution, we present three complementary computational methods to locate ZGV points in generally anisotropic and possibly transversely inhomogeneous waveguides. In principle, the techniques are applicable to waveguides of generic cross-sections, although the first procedure is constrained by the problem size. The current work focuses on homogeneous anisotropic plates. The methods do not require calculating the full dispersion curves in advance. The computations are entirely based on the discretized EVP, as this is generally more robust, very fast and comparably simple to implement~\cite{Gravenkamp2012,Gravenkamp2013a}. The first method is a direct and \emph{globally convergent} method that does not require initial values at the cost of being computationally very expensive. This is the first computational technique capable of guaranteeing to locate all ZGV points as long as the problem is not too large. The second method is based on an iterative approach and is very fast but \emph{locally convergent}; hence, it requires one initial guess for every expected ZGV point. The third method first uses a regularized direct approach to obtain close approximations to ZGV points and then refines the result with our second, iterative method. This technique can be applied to rather large problems and is likely to find all ZGV points. Moreover, it would also be possible to use the first procedure on a coarse discretization to obtain initial values that are then refined with the second one. In this sense, the techniques complement each other.

In the following, we first discuss the modeling of ZGV resonances and their properties in Sec.~\ref{sec:zgv_resonances_in_plates}. This includes a short discussion of the discretization of the continuous problem, as this will be the starting point for the upcoming numerical procedures. The direct solution of the discrete ZGV problem is presented in Sec.~\ref{sec:direct_solution_of_the_two_parameter_eigenvalue_problem}. In contrast to this, Sec.~\ref{sec:fast_locally_convergent_newton_iteration} introduces the mentioned iterative solution procedure. The third and last method is given in Sec.~\ref{sec:method_fixed_rel_distance}, which combines the strength of a globally convergent method with the speed of the iterative one. Finally, a conclusion is given in Sec.~\ref{sec:conclusions_and_outlook}.

\section{ZGV resonances in plates} 
\label{sec:zgv_resonances_in_plates}
The waveguide is a linearly elastic infinite plate as depicted in Fig.~\ref{fig:plate}. It is characterized by its thickness~$h$, mass density~$\rho$ and 4th order stiffness tensor~$\ten{c}$, which we assume homogeneous in the following.

\begin{figure}[tb]
  \centering
  \includegraphics{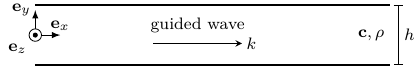}
  \caption{Infinite, elastic, anisotropic plate of thickness~$h$, density~$\rho$, stiffness tensor~$\ten{c}$.}
  \label{fig:plate}
\end{figure}

Displacements~$\ten{\bar{u}}$ in the plate in the absence of external loads are governed by~\cite{langenberg_ultrasonic_2012,kiefer_elastodynamic_2022}
\begin{subequations}
\label{eq:BVP_general}%
\begin{align}
  \nabla \cdot \ten{c} : \nabla \ten{\bar{u}} - \rho \partial_t^2\ten{\bar{u}}  &= \ten{0} \,, \text{ and }
  \label{eq:equation_motion_time} \\
  \dirvec{y} \cdot \ten{c} : \nabla \ten{\bar{u}} &= \ten{0} \,\hphantom{,} \text{ at } y = \pm h/2 \,,
  \label{eq:BC_free}
\end{align}
\end{subequations}
which represent the equations of motion and the traction-free boundary conditions (BC), respectively. Note that the stress tensor $\ten{T} = \ten{c} : \nabla \bar{\ten{u}}$ has already been eliminated therein. $\nabla = \dirvec{x} \partial_x + \dirvec{y} \partial_y + \dirvec{z} \partial_z$ is the Nabla operator, $\dirvec{i}, i \in \{x, y, z\}$, are the unit-directional vectors and the $\partial_i$ denote partial differentiation with respect to the indicated variable. Each ``$\cdot$''-symbol denotes a contraction (scalar product) between adjacent tensor dimensions. Accordingly, the ``:''-symbol implies two consecutive contractions. Note also that $\dirvec{y}$ is the unit-directional vector normal to the plate's surface that yields the relevant traction $\dirvec{y} \cdot \ten{T}$ for the BC. For details on symbolic tensor notation, refer, e.g., to Ref.~\citen{langenberg_ultrasonic_2012}.

\subsection{Guided waves}
\label{sub:guided_waves}
We are interested in time-harmonic, plane, guided waves which cause displacements of the form
\begin{equation}
  \ten{\bar{u}}(x,y,t) = \ten{u}(y, k, \omega) \e^{\iu k x - \iu \omega t} \,,
  \label{eq:plane_wave_ansatz}
\end{equation}
where $k$ denotes the wavenumber along the axial coordinate $x$, $\omega$ stands for the angular frequency in time $t$, and $\iu = \sqrt{-1}$. 
While the dependence on $t$ and $x$ has been resolved analytically by the ansatz (\ref{eq:plane_wave_ansatz}), the dependence on the $y$-coordinate remains to be determined. This is achieved by requiring (\ref{eq:plane_wave_ansatz}) to satisfy (\ref{eq:equation_motion_time}) and (\ref{eq:BC_free}), i.e.,
\begin{subequations}
\label{eq:waveguide_problem}%
\begin{align}
  \label{eq:waveguide_problem_motion}
  \underbrace{ \left[ (\iu k)^2 \ten{c}_{xx} + \iu k (\ten{c}_{xy} + \ten{c}_{yx}) \partial_y + \ten{c}_{yy} \partial_y^2 + \omega^2 \rho \ten{I} \right] }_{\ten{W}} \cdot \ten{u} &= \ten{0} \,, \\
  \label{eq:waveguide_problem_BC}
  \underbrace{ \left[ \iu k \ten{c}_{yx} + \ten{c}_{yy}\partial_y \right] }_{\ten{B}} \cdot \ten{u} &= \ten{0} 
\end{align}
\end{subequations}
which uses the 2nd-order tensors $\ten{c}_{ij} = \dirvec{i} \cdot \ten{c} \cdot \dirvec{j}$. For example, the components of $\ten{c}_{xy}$ are $c_{xkly}$ with $k, l \in \{x, y, z\}$. For the sake of completeness, the explicit derivation of the equations is given in Appendix~\ref{sec:derivation_of_the_waveguide_problem}.
The above boundary value problem in the coordinate $y$ will be referred to as the \emph{waveguide problem} and is considered to depend on the two parameters $\omega$ and~$k$. For future reference, we have identified therein the wave operator $\ten{W}$ and the boundary operator $\ten{B}$. One possible computer implementation of (\ref{eq:waveguide_problem}) based on the spectral collocation method is given in Ref.~\citen{kiefer_gew_2022}.

Solutions of the form (\ref{eq:plane_wave_ansatz}) are denoted as \emph{guided waves} and are well studied in the literature~\cite{auld_acoustic_1990,royer_elastic_2022,kiefer_elastodynamic_2022}. Depending on $\ten{c}$, the displacement components $u_x$ and $u_y$ might decouple from $u_z$. Waves polarized purely in the former plane are denoted as \emph{Lamb waves}, while those in the latter polarization are termed \emph{shear-horizontal (SH) waves}. Note that the waveguide problem formulated in (\ref{eq:waveguide_problem}) is equally valid for all of these polarizations when considering the corresponding 3-, 2- or 1-component tensors, respectively. For example, Lamb waves -- being polarized in $x$-$y$ -- are obtained by removing the $z$-component of all tensors. In any case, the solutions $(\omega, k)$ form curves $\omega(k)$ denoted as \emph{dispersion curves}. An example for Lamb waves in an isotropic steel plate (Lamé parameters\footnote{The isotropic 4th-order stiffness tensor is $\ten{c} = \lambda \ten{I} \ten{I} + \mu (\ten{I} \ten{I}^{1342} + \ten{I} \ten{I}^{1324})$. The super-indices describe a permutation applied to the tensor $\ten{I} \ten{I}$, i.e., the original dimensions $1234$ are re-ordered as described.}
$\lambda = \SI{115.6}{\giga\pascal}$, $\mu = \SI{79}{\giga\pascal}$, mass density $\rho = \SI{7900}{\kilo\gram\per\meter\cubed}$) is shown in Fig.~\ref{fig:dispersion_steel}a. Another important physical property of the waves is their group velocity~$c_\mup{g} \define \frac{\partial \omega}{\partial k}$, which is depicted in Fig.~\ref{fig:dispersion_steel}b for the same plate. Physically, $c_\mup{g}(\omega)$ describes the propagation speed of a pulse whose spectrum is centered at $\omega$.

\begin{figure*}[tb]
  \centering
  \includegraphics{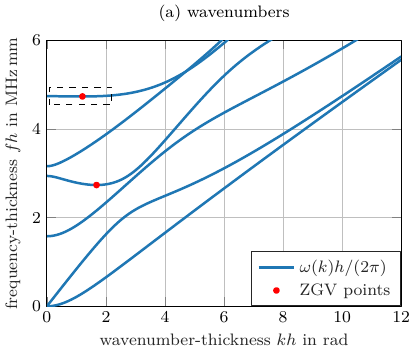}%
  \hspace{1em}
  \includegraphics{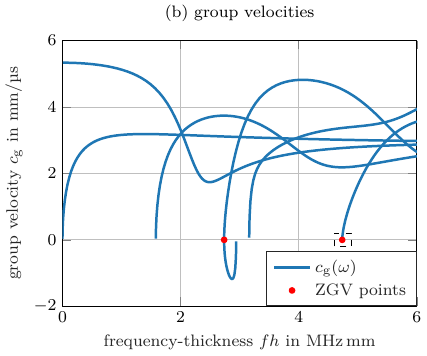}%
  \hspace{1em}%
  \includegraphics{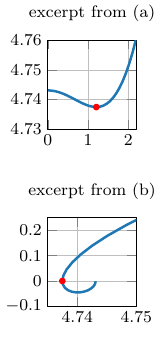}%
  \caption{Lamb wave dispersion curves of an isotropic steel plate showing the frequency-dependent (a) wavenumbers~$k$ and (b) group velocities~$c_\mup{g}$. Excerpts from the dashed regions are also displayed for clarification. Only the real spectrum is shown, for which $c_\mup{g}$ is well defined. The ZGV points are marked.}
  \label{fig:dispersion_steel}
\end{figure*}

\subsection{Properties of ZGV resonances}
\label{sub:properties_zgv_resonances}
We are interested in the ZGV points $(\omega_*, k_*)$ where the group velocity~$c_\mup{g}$ vanishes but $k_*$ remains finite. Two such points appear in Fig.~\ref{fig:dispersion_steel} and are marked correspondingly. In addition, the group velocity usually also vanishes at the cutoff frequencies, i.e., the thickness resonances where $k = 0$, except when multiple cut-off frequencies coincide~\cite{kausel_number_2012}. This means that the points at $k = 0$ are usually local extrema of a dispersion curve, and so are the ZGV points. Note that the latter can also be interpreted as local resonances~\cite{prada_laser-based_2005,prada_local_2008}. Both the thickness resonances and the ZGV resonances lead to long-lasting local vibrations that manifest as peaks in the spectrum of the particle velocity measured at the excitation point.

The parity of the number of ZGV points encountered on a single branch is directly related to whether this branch emerges at $k = 0$ as a forward or a backward wave. When an even (odd) number of ZGV points lie on one branch, they will be located on a forward (backward) wave branch emerging at $k = 0$ (inspect Fig.~\ref{fig:dispersion_steelAustenitic} as an example). This is due to the fact that for large wavenumbers ($k \to \infty$), the group velocity is positive, i.e., $c_\mup{g} = \frac{\partial \omega}{\partial k} > 0$. Indeed, for sufficiently small wavelengths $\frac{2 \pi}{k}$, the system resembles an infinite domain without dispersion. In other words, above a certain value of $k$, the curves $\omega(k)$ increase monotonically. Accordingly, the last ZGV point on a branch (counted in increasing $k$) always corresponds to a minimum in $\omega(k)$. All other ZGV points are alternately a maximum and a minimum.

At points where the group velocity vanishes, the plate admits an additional solution. This has been discussed by Mindlin~\cite{mindlin_vibrations_1958} for the case of thickness resonances. His observation was extended to ZGV points in isotropic and homogeneous plates by Tassoulas and Akylas~\cite{tassoulas_wave_1984}, who denoted the additional solution as \emph{exceptional mode}. Later, Kausel~\cite{kausel_number_2012} showed the existence of these additional solutions in layered isotropic plates with various boundary conditions. The appearance of the exceptional mode implies that these points are double eigenvalue points, i.e., two solutions with the same frequency and wavenumber exist. This is in agreement with the fact that the ZGV points are branching points where two connected real-valued branches for $\omega > \omega_*$ transition into two complex-valued branches for $\omega < \omega_*$ (not shown in Fig.~\ref{fig:dispersion_steel})~\cite{auld_acoustic_1990}.

A further property of ZGV points is that these waves do not propagate energy. For lossless waveguides, the group velocity is equivalent to the energy velocity, defined as the ratio of the wave's power flux vector and its total stored energy~\cite{auld_acoustic_1990,langenberg_ultrasonic_2012}. Hence, their power flux vanishes. In anisotropic plates, the group velocity vector $\ten{c}_\mup{g} = \partial \omega/\partial k\, \dirvec{x} + c_{\mup{g}z} \dirvec{z}$ does not need to be collinear to the wave vector $\ten{k} = k \dirvec{x}$, i.e., the component $c_{\mup{g}z}$ might be non-zero~\cite{shuvalov_backward_2008,auld_acoustic_1990,langenberg_ultrasonic_2012}. However, if the material is invariant to reflection along $\dirvec{z}$, then $c_{\mup{g}z}$ needs to vanish, and $\ten{c}_\mup{g}$ will be aligned with $\ten{k}$, i.e., it is sufficient to consider the $x$-component of $\ten{c}_\mup{g}$. Note that this is the case when Lamb and SH polarized waves decouple (see details in Appendix~\ref{sec:collinearity_of_power_flux_and_wave_vector}), as will be the case for all examples presented in this contribution. Hence, we find ``true ZGV points'' in the sense that the power flux vanishes in all directions. Note that all our methods still work when the polarizations do not decouple. That is to say, points with vanishing power flux in direction $\dirvec{x}$ are correctly found.

Although non-essential to our developments, in the following, we also discuss the normalizability of ZGV resonances for the sake of completeness. As guided waves are usually normalized to unit power flux, this implies that ZGV waves are not normalizable in the usual sense. While a renormalization has successfully been performed for cut-off frequencies~\cite{pagneux_lamb_2006}, this has not been achieved for the ZGV points. On the one hand, this property leads to computational difficulties, e.g., it is not possible to perform a perturbation-based sensitivity analysis~\cite{auld_acoustic_1990,kiefer_transit_2022,kiefer_elastodynamic_2022}. On the other hand, it is precisely this property that makes the waves interesting for practical measurements because the wave's energy remains close to the source.

\subsection{Modeling ZGV resonances}
\label{sub:modeling_zgv_resonances}

While the dispersion curves can be obtained by prescribing values either for the frequencies or for the wavenumbers and then computing the other one, this is not possible for the ZGV points. The reason is that they are isolated points on the dispersion curves. This means that we need to determine both their angular frequency~$\omega_*$ and their wavenumber $k_*$ simultaneously. To this end, an additional equation is required that complements (\ref{eq:waveguide_problem}). 

As such a condition, we demand the appearance of an \emph{exceptional mode}, as discussed in the previous subsection. Adapting slightly from Ref.~\citen{tassoulas_wave_1984}, this solution has the form
\begin{equation}
  \ten{\bar{u}}(x,y,t) = \left[ \ten{u}'(y, k, \omega) + x\, \ten{u}(y, k, \omega) \right] \e^{\iu k x - \iu \omega t} \,,
  \label{eq:exceptional_mode_ansatz}
\end{equation}
where $\ten{u}(y, k, \omega)$ must satisfy the waveguide problem and the dash ($'$) denotes the derivative with respect to $\iu k$. Although the imaginary unit in the derivative is not required, it does lead to simpler expressions.

The equations governing exceptional modes are known for isotropic plates~\cite{tassoulas_wave_1984}. We derive the equations for the anisotropic case by inserting (\ref{eq:exceptional_mode_ansatz}) into (\ref{eq:BVP_general}). While the details can be found in Appendix~\ref{sec:derivation_of_the_exceptional_mode_equation}, the procedure is outlined in the following. Rearranging yields
\begin{equation}
  x \ten{W} \cdot \ten{u} + \ten{W} \cdot \ten{u}' + \ten{W}' \cdot \ten{u} = \ten{0} \,,
\end{equation}
where $\ten{W}(k, \omega)$ is the waveguide operator given in (\ref{eq:waveguide_problem_motion}) and
\begin{equation}\label{eq:Wdash_operator}
    \ten{W}'(k) = 2 \iu k \ten{c}_{xx} + (\ten{c}_{xy} + \ten{c}_{yx}) \partial_y \,.
\end{equation}
As $\ten{u}$ was required to be a waveguide solution, it satisfies $\ten{W}\cdot \ten{u} = \ten{0}$, and we may restate the problem as 
\begin{subequations}
\label{eq:exceptional_mode_eq}%
  \begin{align}
    \label{eq:exceptional_mode_motion}
    &\text{PDE:}\quad
    &\begin{bmatrix}
      \ten{W}  &   \ten{0} \\
      \ten{W}' &   \ten{W}
    \end{bmatrix} \cdot \begin{bmatrix}
      \ten{u} \\ \ten{u}'
    \end{bmatrix} &= \begin{bmatrix}
      \ten{0} \\ \ten{0}
    \end{bmatrix} \,, \\
    \label{eq:exceptional_mode_BC}
    &\text{BC:}\quad
    &\begin{bmatrix}
      \ten{B}   &  \ten{0} \\ \ten{B}'   &  \ten{B}
    \end{bmatrix} \cdot \begin{bmatrix}
      \ten{u} \\ \ten{u}'
    \end{bmatrix} &= \begin{bmatrix}
        \ten{0} \\ \ten{0}
      \end{bmatrix} \,,
  \end{align}
\end{subequations}
where the boundary condition (\ref{eq:exceptional_mode_BC}) has been obtained in a similar fashion. Therein, $\ten{B}(k)$ is given in (\ref{eq:waveguide_problem_BC}) and $\ten{B}' = \ten{c}_{yx}$. 

It is important to note that the second equation in the systems (\ref{eq:exceptional_mode_motion}) and (\ref{eq:exceptional_mode_BC}) could have been obtained by differentiating the waveguide problem (\ref{eq:waveguide_problem}) with respect to $\iu k$ and setting $c_\mup{g} = \partial \omega / \partial k = 0$. This demonstrates that a wave field of the form (\ref{eq:exceptional_mode_ansatz}) emerges precisely at the points where the group velocity vanishes.

We seek eigensolutions where (\ref{eq:waveguide_problem}) and (\ref{eq:exceptional_mode_eq}) are simultaneously satisfied. The combination of these two equations may be denoted as the \emph{ZGV problem}. We interpret it as a differential \emph{two-parameter eigenvalue problem (two-parameter EVP)}~\cite{atkinson_multiparameter_2010,plestenjak_spectral_2015,muhic_singular_2009,hochstenbach_solving_2019} that describes the sought isolated points on the dispersion curves. Note that although (\ref{eq:waveguide_problem}) is included in (\ref{eq:exceptional_mode_eq}), we still require a system of both equations as (\ref{eq:waveguide_problem}) guarantees that $\ten{u}\ne 0$. In the following, we introduce a discrete approximation, and subsequently, we present a numerical method to solve the resulting \emph{algebraic two-parameter EVP}.

\subsection{Discretization} 
\label{sub:discretization}
In order to actually solve for the ZGV points, we perform a numerical discretization of the previously discussed boundary-value problems. This converts the \emph{waveguide problem} (\ref{eq:waveguide_problem}) into an algebraic EVP in terms of the vector%
\footnote{Note that we switch from describing physical tensor fields to a more abstract finite-dimensional vector space. Correspondingly, we also switch from tensor notation (bold) to matrix/vector notation (uppercase/lowercase, respectively).}
of unknowns $u$ that might be written as
\begin{equation}
  \label{eq:waveguide_problem_discrete}
  W(k, \omega) u = 0 \,,
\end{equation}
where the parameterized $n \times n$-matrix $W$ defines the discrete waveguide operator and is given by
\begin{equation}\label{eq:waveguide_operator_discrete}
    W(k, \omega) = (\iu k)^2 L_2 + \iu k L_1 + L_0 + \omega^2 M \,,
\end{equation}
with real matrices $L_i$, $M$. Therein, the boundary conditions are already accounted for by the matrices. Similarly, the discretization of (\ref{eq:Wdash_operator}) leads to
\begin{equation}\label{eq:Wdash_operator_discrete}
    W'(k) = 2 \iu k L_2 + L_1 \,.
\end{equation}
The matrices $W$ and $W'$ also define the discrete analogue of the \emph{exceptional mode equation} (\ref{eq:exceptional_mode_eq}), namely, 
\begin{equation}\label{eq:exceptional_mode_eq_discrete}
    \begin{bmatrix}
        W   &   0 \\
        W'  &   W
    \end{bmatrix} \begin{bmatrix}
        u \\ u'
    \end{bmatrix} = \begin{bmatrix}
        0 \\ 0
    \end{bmatrix} \,,
\end{equation}
in terms of the unknown eigenvector $[u\transp, {u'}\transp]\transp$. The boundary conditions are, again, already accounted for by the matrices.

Various numerical discretization methods can be used to obtain the above linear systems. Two popular methods are finite elements~\cite{gravenkamp_efficient_2018,Gravenkamp2013a,Gravenkamp2012,finnveden_evaluation_2004} and spectral collocation~\cite{hernando_quintanilla_modeling_2015,kiefer_elastodynamic_2022,kiefer_calculating_2019,kiefer_gew_2022}. The \emph{spectral element method}, i.e., high-order finite elements \cite{fichtner_full_2010,gravenkamp_efficient_2018,gravenkamp_high-order_2021}, is used in this contribution as it produces matrices of small size $n \times n$ and leads to $L_2$ and $M$ being nonsingular. These properties are highly advantageous -- in some cases even necessary -- in order to successfully use the computational methods presented in the following. The small matrix size is required for the direct solution method because it blows up the matrices' dimensions to  $4 n^2\times 4 n^2$. A further advantage of the spectral element method is that it preserves the symmetry of the continuous operators, resulting in Hermitian matrices $W$ and $\iu W'$ when $k$ is real-valued, i.e., the matrices $L_i$ are alternately symmetric/anti-symmetric and $M$ is symmetric positive definite. As a consequence, solving for the eigenpair $(\omega^2, u)$ of a guided wave at a given real-valued $k$ is a Hermitian problem and the complex conjugate and transpose $u^\mup{H}$ is known to be the left eigenvector corresponding to $u$, i.e., $u^\mup{H} W = 0$. This avoids the explicit computation of the left eigenvectors, which will be exploited when computing group velocities and also to design an extremely fast locally convergent method in Sec.~\ref{sec:fast_locally_convergent_newton_iteration}.

For future reference, we also discuss how to compute the group velocity~$c_\mup{g}$. It is well known that it can directly be computed from the discrete system (\ref{eq:waveguide_problem_discrete})~\cite{finnveden_evaluation_2004,Gravenkamp2013a}. To this end, we differentiate the equation with respect to $\iu k$ and obtain 
\begin{equation}
  \label{eq:waveguide_problem_differentiated}
  [ 2 \iu k L_2 + L_1 + 2 \omega \omega' M ] u + W u' = 0 \,.
\end{equation}
Exploiting that $u^\mup{H}$ is the left eigenvector as discussed above, the unknown $u'$ can be eliminated by multiplying the expression from the left by $u^\mup{H}$. After re-arranging, this yields the group velocity ($x$-component) as
\begin{equation}
  \label{eq:group_velocity_discrete}
  c_\mup{g} = \iu \omega' = -\frac{u^\mup{H} \iu (2 \iu k L_2 + L_1) u}{2\omega\, u^\mup{H} M u} \,.
\end{equation}
Accordingly, at ZGV points, the condition 
\begin{equation}
  \label{eq:ZGV_condition_cg0}
  u^\mup{H} \iu W' u = u^\mup{H} \iu (2 \iu k L_2 + L_1) u = 0 \,
\end{equation}
holds.

\section{Direct solution of the two-parameter eigenvalue problem} 
\label{sec:direct_solution_of_the_two_parameter_eigenvalue_problem}
As discussed in Sec.~\ref{sub:modeling_zgv_resonances}, ZGV resonances are modeled by two coupled EVPs parametrized in $\omega$ and $k$ that need to be satisfied simultaneously, namely the guided wave problem and the exceptional mode problem. The discrete form of this so-called two-parameter EVP was given in Sec.~\ref{sub:discretization}. Here, we present a direct numerical solution procedure.

Explicitly writing out the exceptional mode equation (\ref{eq:exceptional_mode_eq_discrete}) in terms of the matrices $L_i$, $M$ and grouping according to the dependence on $\iu k$ and $\omega^2$ yields
\begin{equation}
  \label{eq:exceptional_mode_discrete_2_block}
  \left[ (\iu k)^2 \widetilde L_2 + \iu k \widetilde L_1 + \widetilde L_0 + \omega^2 \widetilde M \right] \widetilde u = 0 \,,
\end{equation}
with $2n\times 2n$ matrices and vectors
\[
  \widetilde L_2 = \begin{bmatrix} L_2 & 0 \\ 0 & L_2 \end{bmatrix},\
  \widetilde L_1 = \begin{bmatrix} L_1 & 0 \\ 2L_2 & L_1 \end{bmatrix},\
  \widetilde L_0 = \begin{bmatrix} L_0 & 0 \\ L_1 & L_0 \end{bmatrix}, 
\]
\[
\widetilde M = \begin{bmatrix} M & 0 \\ 0 & M \end{bmatrix},\ 
 \widetilde u = \begin{bmatrix} u \\ u'  \end{bmatrix} \,.
\]
At first sight, demanding the appearance of an exceptional mode by requiring (\ref{eq:exceptional_mode_discrete_2_block}) to be satisfied seems more complicated than the ZGV condition (\ref{eq:ZGV_condition_cg0}). However, while (\ref{eq:ZGV_condition_cg0}) is nonlinear in $u$, (\ref{eq:exceptional_mode_discrete_2_block}) is linear in the new unknown $\widetilde u$. This advantage is exploited in the following to design a direct, globally convergent computational method to locate the ZGV points.

Overall, equations \eqref{eq:waveguide_problem_discrete} and \eqref{eq:exceptional_mode_discrete_2_block} form the two-parameter EVP describing ZGV resonances, and we note that it is quadratic in $\iu k$. 
Following Ref.~\citen{HMP_LinQ2MEP}
we introduce $\lambda=\iu k$, $\mu=\omega^2$, $\eta=(\iu k)^2$ and write \eqref{eq:waveguide_problem_discrete} and \eqref{eq:exceptional_mode_discrete_2_block} as a linear three-parameter EVP (see also Appendix~\ref{sec:mep})
\begin{subequations}
\label{eq:3EP_alg1}
\begin{align}
    (\eta L_2 + \lambda L_1 + L_0 + \mu M)u &=0\label{eq:3ep1}\\
    (\eta \widetilde L_2 + \lambda \widetilde L_1 + \widetilde L_0 + \mu \widetilde M)v &=0\label{eq:3ep2}\\
    (\eta C_2 + \lambda C_1 + C_0 \phantom{+ \mu C_2}) w&=0,\label{eq:3ep3}
\end{align}
\end{subequations}
where
\[
  C_2 = \begin{bmatrix} 1 & 0 \\ 0& 0 \end{bmatrix},\ 
  C_1 = \begin{bmatrix} 0 & -1 \\ -1 & 0 \end{bmatrix},\
  C_0 = \begin{bmatrix} 0 & 0 \\ 0 & 1 \end{bmatrix}.
\]
Note that equation \eqref{eq:3ep3} incorporates the relation between $\lambda$ and $\eta$ since $\det(\eta C_2 + \lambda C_1 + C_0)=\eta-\lambda^2$. We remark that it would also have been possible to linearize via companion linearization~\cite{TisseurMeerbergen}, which would have doubled the problem size. In contrast to this, the above approach only increases the problem size by two at the expense of introducing an additional parameter.

The three-parameter EVP is related to a system of conventional generalized eigenvalue problems (GEPs) that decouple in the eigenvalues but remain coupled through their common eigenvector~\cite{plestenjak_spectral_2015}. These GEPs are 
\begin{equation}\label{eq:delta_sistem}
    \Delta_1 z = \lambda \Delta_0 z,\ \
    \Delta_M z = \mu \Delta_0 z,\  \
    \Delta_2 z = \eta \Delta_0 z,
\end{equation}
where the new eigenvector $z$ is given by the Kronecker product $z=u\otimes v\otimes w$ and the $4n^2\times 4n^2$ matrices
\begin{equation}\label{eq:delta01}
    \Delta_0 = \left|\begin{matrix}
       L_2 & L_1 & M \cr
       \widetilde L_2 & \widetilde L_1 & \widetilde M \cr
       C_2 & C_1 & 0 
    \end{matrix}\right|_\otimes,\
    \Delta_1= (-1)\left|\begin{matrix}
       L_2 & L_0 & M\cr
       \widetilde L_2 & \widetilde L_0 & \widetilde M \cr
       C_2 & C_0 & 0
    \end{matrix}\right|_\otimes    
\end{equation}
are the so-called operator determinants computed using the Kronecker product. %
\footnote{{The Kronecker product $A\otimes B$ of the $m \times n$-matrix $A = [A_{ij}]$ with the $p \times q$-matrix $B$ yields the block matrix $[A_{ij} B]$ of size $m p \times n q$. For three
matrices, $A\otimes B\otimes C=(A\otimes B)\otimes C = A\otimes (B\otimes C)$. Vectors are treated like matrices with one column. The operator determinants are computed analogously to determinants of matrices applying the Kronecker product instead of multiplication of scalar matrix entries. For example, $\Delta_0 =  L_1 \otimes \widetilde M \otimes C_2 + M \otimes \widetilde L_2 \otimes C_1 - M \otimes \widetilde L_1 \otimes C_2 - L_2 \otimes \widetilde M \otimes C_1$.}} 
$\Delta_2,\Delta_M$ are defined in a similar way but are not needed in the following. 
The three-parameter EVP \eqref{eq:3EP_alg1} is \emph{singular} which
means that all GEPs in \eqref{eq:delta_sistem} are \emph{singular}, i.e., $\det(\Delta_1-\lambda \Delta_0)\equiv 0$, but the problem has a finite number of eigenvalues. 

In Theorem \ref{thm:main}, we show that if $(\omega_*, k_*)$ is a ZGV point, then $\lambda_*=\iu k_*$ is an eigenvalue of the first GEP in (\ref{eq:delta_sistem}). 
We can thus compute the  wavenumbers $k_*$ at the ZGV points  by solving the GEP $\Delta_1 z = \lambda \Delta_0 z$ for $k = -\iu\lambda$. These values are subsequently substituted into (\ref{eq:waveguide_problem_discrete}) to retrieve a set of corresponding angular frequencies $\omega$, which, in addition, require to satisfy Eq.~\eqref{eq:ZGV_condition_cg0} for the solution to represent a ZGV point. To find the finite eigenvalues of the singular GEP, we apply the rank projection algorithm described in Ref.~\citen{HMP22_arxiv} and implemented in \texttt{MultiParEig}~\cite{multipareig_2023}. In addition to ZGV points, the three-parameter EVP has additional eigenvalues. To extract ZGV points, we select solutions where $\lambda = \iu k$ is strictly imaginary, $\mu = \omega^2$ is real, and \eqref{eq:ZGV_condition_cg0} holds. The algorithm is the following.

{\small%
\noindent\vrule height 0pt depth 0.5pt width \linewidth \\*
\textbf{Algorithm~1: direct method for ZGV points} \\
\textbf{Input:} $n\times n$ matrices $L_2, L_1, L_0, M$ \\
\textbf{Output:} ZGV points $(\omega_*,k_*)$\\[-1.2ex]
\vrule height 0pt depth 0.3pt width \linewidth \\*[2pt]
\begin{tabular}{ll}
 {\footnotesize 1:} & build matrices $\Delta_0$ and $\Delta_1$ in \eqref{eq:delta01}\cr
 {\footnotesize 2:} & solve the singular GEP $\Delta_1 z = \lambda \Delta_0 z$\cr
 {\footnotesize 3:} & for each imaginary $\iu k_j=\lambda_j$, $j=1,\ldots,r$\cr
 {\footnotesize 4:} & \hbox{}\quad solve $\left[ (\iu k_j)^2 L_2 + (\iu k_j) L_1 + L_0 + \omega^2 M \right] u = 0$\cr
 {\footnotesize 5:} & \hbox{}\quad for each $\omega_\ell$, $\ell=1,\ldots,n $\cr
 {\footnotesize 6:} & \hbox{}\quad\quad return $(\omega_\ell,k_j)$ if \eqref{eq:ZGV_condition_cg0} holds
\end{tabular}\\*
\vrule height 0pt depth 0.5pt width \linewidth \\
}

\begin{figure*}[bt]
  \centering
  \includegraphics{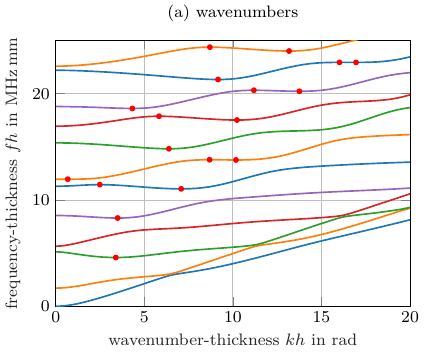}%
  \hspace{2em}
  \includegraphics{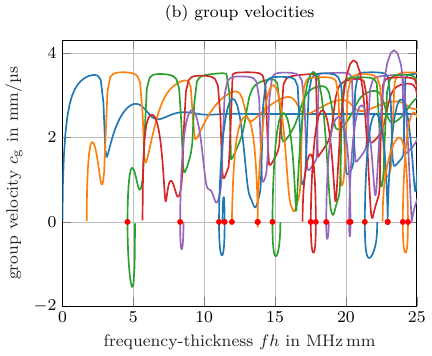}
  \caption{Antisymmetric Lamb waves in an austenitic steel plate: (a) wavenumbers~$k$, (b) group velocities~$c_\mup{g}$.}
  \label{fig:dispersion_steelAustenitic}
\end{figure*}

\noindent
To demonstrate the reliability of the developed method, we compute the ZGV points of antisymmetric Lamb waves in an austenitic steel plate. The orthotropic stiffness tensor is given in Appendix~\ref{sec:material_parameters}. Figure~\ref{fig:dispersion_steelAustenitic} shows the resulting wavenumber and group velocity dispersion curves. All \num{18} ZGV points were obtained and are marked therein. The used matrices $L_i, M$ were of size $\num{39}\times \num{39}$ (as will be for all other computations involving this example), leading to matrices $\Delta_i$ of size $\num{6084}\times \num{6084}$. The computational time was \SI{381}{\second} (Intel Core i9, 16 GB RAM, as for all upcoming numerical experiments). Note that the method is able to reliably locate all ZGV points, including double and triple ZGV points on a single branch. It is remarkable that no initial values are required in doing so. This is a strong advantage over other computational techniques.

Nonetheless, its high computational cost is an important drawback of the presented direct method. While the original waveguide problem is of size $n \times n$, the ZGV calculation requires to \emph{additionally} solve a singular GEP of size $4n^2\times 4n^2$. This becomes unfeasible very quickly as all current numerical solvers for singular GEPs are direct methods in the sense that they compute all eigenvalues of a GEP at once. As a result, in practice, we can only apply the above method to single-layered plates. For this reason,  we present in the following a rapidly converging iterative method formulated in a very general way and applicable to large problems.

\section{Fast locally convergent Newton-type iteration}
\label{sec:fast_locally_convergent_newton_iteration}
A different approach is taken in the following. 
In terms of the $n + 2$ unknowns 
\[
  p = \begin{bmatrix}
    u \\ k \\ \omega^2 
  \end{bmatrix} \,,
\]
we can define a $(n+2) \times (n+2)$ nonlinear system of the form $F(p) = 0$ that describes the ZGV points. We write out $F(p)$ using (\ref{eq:waveguide_problem_discrete}) and (\ref{eq:ZGV_condition_cg0}), which yields
\begin{subequations}
\label{eq:ZGV_system_nonlinear}
\begin{align}
   W u = \left[ k^2 (-L_2) + k (\iu L_1) + L_0 + \omega^2 M \right] u &= 0 \,, \label{eq:Fp_a}\\
   u^\mup{H} \iu W' u = u^\mup{H} \left[ k (-2 L_2) + \iu L_1 \right] u &= 0 \,, \label{eq:Fp_b}\\
   u^\mup{H} u - 1 &= 0 \label{eq:Fp_c}\,.
\end{align}
\end{subequations}
This represents a so-called two-dimensional EVP~\cite{lu_newton-type_2022}. Contrary to (\ref{eq:3EP_alg1}), it consists only of one EVP with an additional scalar nonlinear constraint. The goal is to find a root $p_*$ of (\ref{eq:ZGV_system_nonlinear}) given an initial guess $p_0$. To solve for $p_*$, we cannot apply the Newton method directly since $u$ is, in general, a complex vector and \eqref{eq:Fp_b} and 
\eqref{eq:Fp_c} are not complex differentiable in $u$ due to the presence of conjugate values.

To overcome the problem of complex differentiability, we employ a Newton-type iteration from Ref.~\citen{lu_newton-type_2022} derived for a very 
similar problem.
The idea is to use 
\begin{equation}
  \label{eq:jacobian}
  J(p) = \begin{bmatrix}
    W   &   \iu W' u   &    M u \\
    2 u^\mup{H} \iu W'   &   -u^\mup{H} 2 L_2 u    &   0 \\
    2 u^\mup{H}     &   0    &   0
  \end{bmatrix}
\end{equation}
as a natural complex extension of the Jacobian with respect to $p$ of (\ref{eq:ZGV_system_nonlinear}).
This exploits the fact that both $W$ and $\iu W'$ are Hermitian.
As proposed in Ref.~\citen{lu_newton-type_2022}, we obtain the update $\Delta p_j$ for $p_{j+1} = p_{j} + \Delta p_j$ as
\begin{equation}\label{eq:updated_newton}
  \Delta p_j = -\underbrace{J(p_j)^{-1}F(p_j)}_{q_j} +\iu \beta_j \underbrace{J(p_j)^{-1}e_{n+1}}_{s_j}\,,
\end{equation}
where $e_{n+1}$ is the canonical unit vector, i.e., $e_{n+1}$ is an ${(n+2)}$-dimensional vector whose $(n+1)^\text{th}$ component equals 1, and all other components are equal to~0. Furthermore, we take $\beta_j=\Im(q_{j,n+1})/\Re(s_{j,n+1})$. The correction $\beta_j$ ensures that $k_{j+1}$ remains real (we assume that $k_j$ and $\omega^2_j$ are real).

If we do not have an initial approximation for $u$, usually a good choice is to take the right singular vector corresponding to the smallest singular value of $W(k_0, \omega_0)$. We observe that in practice, this choice of $u_0$ leads to almost zero $\beta_j$. While in this case, the $\beta$-correction is not necessary, it is important for the convergence of the algorithm with other choices of $u_0$.

One way to find all ZGV points of Fig.~\ref{fig:dispersion_steelAustenitic} using the Newton-type method is to explicitly compute the group-velocity dispersion curves $c_\mup{g}(\omega)$,
locate the changes in sign of $c_\mup{g}(\omega)$ and use the corresponding values of $\omega$ and $k$ as initial guesses for our algorithm. 
This requires computing the dispersion curves with a sufficiently fine resolution in the prescribed $k$-values, such that at least one sample ($k$-value) is on each of the backward-wave branches. Otherwise, ZGV points will be missed. Using \num{400} wavenumber points and applying this approach to the austenitic steel plate leads to Fig.~\ref{fig:dispersion_steelAustenitic} in \SI{0.62}{\second}. The time to compute the required dispersion curves and group velocity is included therein.

Only if the initial guess is close enough to the desired ZGV point will the algorithm converge correctly. The regions of convergence for the anti-symmetric Lamb waves in the austenitic steel plate are displayed in Fig.~\ref{fig:regions_of_convergence}. Therein, each pixel defines an initial guess $(\omega_0, k_0)$, while $u_0$ is always taken as the right singular vector as described above. If the algorithm converged to any of the predefined ZGV points marked in the figure, the pixel is classified according to which ZGV point it converged to. This gives the colored regions of convergence in the figure, where gray indicates convergence to a cut-off frequency. White pixels correspond to starting values that did not converge to any of the ZGV points or cut-off frequencies. We observe large regions of convergence, allowing us to reliably use the algorithm even with rather poor initial guesses for $\omega$ and $k$. Critical are the situations where multiple ZGV points are very close or when they are close to the cutoff. 

\begin{figure}[tb]
  \centering
  \includegraphics{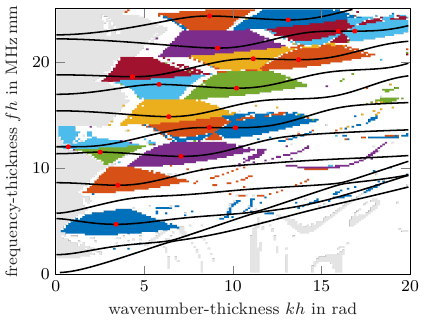}
  \caption{Regions of convergence: initial guesses $\omega_0 = 2 \pi f_0$ and $k_0$ that are close enough to a ZGV point converge correctly. Each colored pixel indicates convergence of that initial guess towards a corresponding ZGV point (gray pixels to cut-off frequencies).}
  \label{fig:regions_of_convergence}
\end{figure}

Our choice of the generalized Jacobian given in (\ref{eq:jacobian}) together with the $\beta$-correction leads to a significantly faster algorithm. Instead, in order to avoid the problem of complex-differentiability, $u$ and $u^\mup{H}$ would usually be treated as two different unknowns, leading to a considerably larger parameter vector $p$ and associated Jacobian. To obtain Fig.~\ref{fig:regions_of_convergence} that shows the result of $\num{150}\times\num{150}$ initial guesses, our method performs \num{190101} iteration steps in \SI{35.1}{\second}. The corresponding result of the common Newton method looks almost identical but requires \num{190289} iteration steps in \SI{120}{\second}.

An analysis of the convergence behavior with the iteration steps is depicted in Fig.~\ref{fig:convergence_vs_iterations}. The mean, as well as the minimum and maximum relative error in frequency for the \num{100} tested initial guesses, is shown. Numerical accuracy is achieved after only \num{5} iteration steps, underlining the efficiency of the method.

\begin{figure}[tb]
    \centering
    \includegraphics{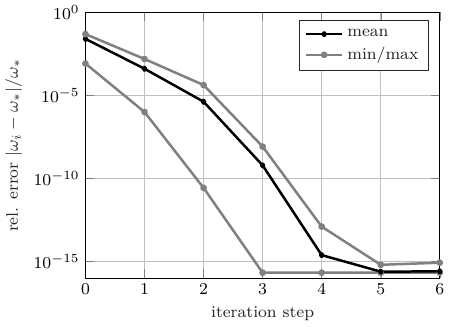}
    \caption{Decrease of the relative error in frequency with each iteration step of the Newton-type method for the ZGV point at $(\omega_* = 2 \pi \cdot \SI{11}{\mega\hertz}, k_* = \SI{7.1}{\radian\per\milli\meter})$. The \num{100} initial guesses are randomly picked from the uniform distribution with $\SI{\pm 5}{\percent}$ variation around the ZGV point, which corresponds to the convergence radius. Other ZGV points behave similarly.}
    \label{fig:convergence_vs_iterations}
\end{figure}

\section{Method of fixed relative distance}
\label{sec:method_fixed_rel_distance}
There are currently no efficient methods for large singular GEPs, and Algorithm 1 cannot be applied to problems with large $n$. We can apply the Newton-type iteration from Sec.~\ref{sec:fast_locally_convergent_newton_iteration}, but this method requires appropriate initial approximations. In this section, we present a method suitable for large problems that can compute a small number of ZGV points $(\omega_*,k_*)$ such that 
$k_*$ is close to a given target $k_0$.

As discussed previously, at a ZGV frequency~$\omega_*$, the waveguide problem in the variables $\lambda = \iu k$ and $\mu_* = \omega_*^2$, i.e.,
\begin{equation}
  \label{eq:waveguide_problem_discrete_pert}
  [ \lambda^2 L_2 + \lambda L_1 + L_0 + \mu_* M ] u = 0 \,,
\end{equation}
has a multiple (usually double) eigenvalue~$\lambda_* = \iu k_*$. Note that this follows from \eqref{eq:waveguide_problem_discrete} and \eqref{eq:exceptional_mode_discrete_2_block} both holding at the ZGV point.
Therefore, for certain $\widetilde\mu \ne \mu_*$ but close to $\mu_*$,
\begin{equation}\label{eq:pert_QEP}
[\lambda^2 L_2 + \lambda L_1 + L_0 + \widetilde\mu M ] u=0
\end{equation}
has \emph{two different eigenvalues} that are both close to $\lambda_*$. The method of fixed relative distance (MFRD) can be employed to find such points~\cite{jarlebring_computing_2011}. Adapted to our problem, it introduces the three-parameter EVP 
\begin{subequations}
\label{eq:3ep_pert}
\begin{align}
    (\eta L_2 + \lambda L_1 + L_0 + \mu M )u &=0\label{eq:3ep1pert}\\
    (\eta (1+\delta)^2 L_2 + \lambda (1+\delta) L_1 + L_0 + \mu M)v &=0\label{eq:3ep2pert}\\
    (\eta C_2 + \lambda C_1 + C_0 \phantom{+ \mu C_2}) w&=0\label{eq:3ep3pert}
\end{align}
\end{subequations}
in $\mu = \omega^2$, $\lambda = \iu k$ and $\eta = (\iu k)^2$. Therein, $\delta>0$ specifies the relative distance between the sought $\lambda$ and serves as a regularization parameter. Moreover, the matrices $C_0,C_1,C_2$ are as in \eqref{eq:3ep3}.
We conclude that for small $\delta$, the three-parameter EVP \eqref{eq:3ep_pert} has an eigenvalue
$(\widetilde\lambda,\widetilde \mu,\widetilde\lambda^2)$ close to the ZGV point $(\lambda_*,\mu_*,\lambda_*^2)$ such that $\widetilde\lambda$ and $\widetilde\lambda(1+\delta)$ are eigenvalues of the initial problem~\eqref{eq:pert_QEP}.

Solutions are obtained similarly to Sec.~\ref{sec:direct_solution_of_the_two_parameter_eigenvalue_problem} by performing a transformation into a system of conventional GEPs. However, in contrast to (\ref{eq:3EP_alg1}), the three-parameter EVP \eqref{eq:3ep_pert} is regular since
the corresponding $2n^2\times 2n^2$ matrix
\begin{equation}
    \widetilde \Delta_0 = \left|\begin{matrix}
       L_2 & L_1 & M\cr
       (1+\delta)^2 L_2 & (1+\delta)L_1 & M \cr
       C_2 & C_1 & 0 
    \end{matrix}\right|_\otimes\label{eq:threeeig_pert_1}
\end{equation}   
is nonsingular for $\delta>0$. Hence, the first GEP given by
\begin{equation}\label{eq:Delta_pert_3par}
    \widetilde \Delta_1z = \lambda \widetilde \Delta_0 z,
\end{equation} 
where
\begin{equation}
    \widetilde \Delta_1 = (-1)\left|\begin{matrix}
       L_2 & L_0 & M\cr
       (1+\delta)^2 L_2 & L_0 & M \cr
       C_2 & C_0 & 0
    \end{matrix}\right|_\otimes,\label{eq:threeeig_pert_2}
\end{equation}
is also regular and we can apply standard subspace methods,
for instance \verb|eigs| in Matlab, to compute some solutions $\lambda$ close to a chosen target $\lambda_0 = \iu k_0$. Then, the obtained eigenvector $z$ is used in the GEP associated with $\mu$, namely, 
\begin{equation}\label{eq:Delta_pert_3par_mu}
    \widetilde \Delta_M z = \mu \widetilde \Delta_0 z,
\end{equation} 
with $\widetilde\Delta_0$ as before and
\begin{align}
    \widetilde \Delta_M&= (-1) \left|\begin{matrix}
       L_2 & L_1 & L_0 \cr
       (1+\delta)^2 L_2 & (1+\delta)L_1 & L_0 \cr
       C_2 & C_1 & C_0 
    \end{matrix}\right|_\otimes, \nonumber
\end{align}
to obtain $\mu$ via
\begin{equation}
  \mu = \frac{z^\mup{H} \widetilde\Delta_M z}{z^\mup{H} \widetilde\Delta_0 z} \,.
\end{equation}
Since each solution $(\lambda, \mu)$ is close to a ZGV point, it can be used as an initial approximation for the Newton-type method from Sec.~\ref{sec:fast_locally_convergent_newton_iteration}.
The algorithm is the following.

{\small%
\noindent\vrule height 0pt depth 0.5pt width \linewidth \\*
\textbf{Algorithm~2: MFRD method for ZGV points} \\
\textbf{Input:} $n\times n$ matrices $L_2, L_1, L_0, M$, target $k_0$ \\
\textbf{Output:} ZGV points $(\omega_*, k_*)$ close to $k_0$\\*[-1.2ex]
\vrule height 0pt depth 0.3pt width \linewidth \\[2pt]
\begin{tabular}{ll}
 {\footnotesize 1:} & build matrices $\widetilde\Delta_0$ and $\widetilde \Delta_1$ in \eqref{eq:threeeig_pert_1} and \eqref{eq:threeeig_pert_2}\cr
 {\footnotesize 2:} & find eigenvalues of $\widetilde\Delta_1z=\lambda \widetilde\Delta_0z$ close to $\lambda_0 = \iu k_0$\cr
 {\footnotesize 3:} & for each $\lambda$ and eigenvector $z$\cr
 {\footnotesize 4:} & \hbox{}\quad compute $\mu=z^\mup{H}\widetilde \Delta_M z /z^\mup{H} \widetilde \Delta_0 z$\cr
 {\footnotesize 5:} & \hbox{}\quad if $|\Re(\lambda)|$ and $|\Im(\mu)|$ are both small then \cr
 {\footnotesize 6:} & \hbox{}\quad\quad apply Newton-type method from Sec.~\ref{sec:fast_locally_convergent_newton_iteration} to \cr
 & \hbox{}\quad\quad compute $(k_*,\omega_*,u_*)$ with initial guess $\Im(\lambda),\Re(\mu)$\cr 
 {\footnotesize 7:} & \hbox{}\quad\quad return $(\omega_*,k_*)$ if \eqref{eq:ZGV_condition_cg0} holds
\end{tabular}\\*
\vrule height 0pt depth 0.5pt width \linewidth \\
}

\begin{table*}[bt]\centering
    \caption{Comparison of the algorithms and their properties. The indicated computational time is for the antisymmetric Lamb waves in the austenitic steel plate of Fig.~\ref{fig:dispersion_steelAustenitic}.}
    \label{tab:comparison_algorithms}
    \setlength\tabcolsep{2.3ex}
    \begin{tabular}{||l c c c c||}
        \hline
        \textbf{algorithm} & \textbf{problem size} & \textbf{initial guesses} & \textbf{finds all solutions} & \textbf{speed} \\
        \hline\hline
        direct method & small & no & yes & very slow (\SI{381}{\second}) \\
        \hline
        Newton-type & large & yes & likely (using $c_\mup{g}(\omega)$) & very fast (\SI{0.62}{\second})  \\
        \hline
        MFRD+Newton-type & medium & target $k_0$ & very likely & fast (\SI{18}{\second}) \\
        \hline
    \end{tabular}
\end{table*}

\noindent
We can apply Algorithm 2 several times using different targets $k_0$ and thus scan a wavenumber interval $[k_a,k_b]$ for ZGV points $(\omega_*,k_*)$. It is worth remarking that from one search to the next, the target $k_0$ can be chosen such that it becomes highly unlikely to miss any ZGV point. To this end, $k_0$ is chosen close to the largest $k_*$ found in the previous search (at a smaller $k_0$). Alternatively, if $n$ is small enough, we can compute all eigenvalues in step 2, refine the solutions with the Newton-type method and thus obtain all ZGV points. An advantage of 
Algorithm 2 over Algorithm 1 is that the GEP \eqref{eq:Delta_pert_3par} is of size $2n^2\times 2n^2$ instead of $4n^2\times 4n^2$. But most importantly, the problem is regular if the employed numerical discretization method provides regular matrices.

The regularization parameter $\delta$ has to be selected carefully. If $\delta$ is too small, then the GEP \eqref{eq:Delta_pert_3par} is close to being singular and the method of choice might fail to find eigenvalues in step 2.
On the other hand, if $\delta$ is too large, then initial approximations might be too poor, and the Newton-type method does not converge.

Applying Algorithm 2 to the same problem as before using $\delta = \num{1e-6}$ yields all ZGV points as already depicted in Fig.~\ref{fig:dispersion_steelAustenitic}. For this end, the computation is done at \num{12} different targets $k_0$. The total computing time is \SI{18}{\second}.

The relation between the two-parameter EVPs given in (\ref{eq:3ep_pert}) and (\ref{eq:3EP_alg1}) should be discussed. As ZGV points are double eigenvalues $\lambda_* = \iu k_*$ for $\mu_* = \omega_*^2$, it might seem appropriate at first sight to set $\delta = 0$ in (\ref{eq:3ep_pert}) to find these solutions. The resulting singular EVP could be solved as in Sec.~\ref{sec:direct_solution_of_the_two_parameter_eigenvalue_problem}. However, this approach is able to find only double eigenvalues of geometric 
multiplicity two~\cite{muhic_method_2014}, i.e., crossing points in the dispersion curves. The geometric multiplicity of an eigenvalue is the number of linearly independent eigenvectors associated with it. ZGV points are double eigenvalues of geometric multiplicity one. For this reason, the extended two-parameter EVP as given in (\ref{eq:3EP_alg1}) is required in order to compute ZGV points with a direct approach.

\section{Conclusions and outlook} 
\label{sec:conclusions_and_outlook}

ZGV points are double eigenvalue points on the dispersion curves that are characterized by the appearance of an exceptional mode (in contrast to crossing points of dispersion curves). We have derived the associated equations for anisotropic plates. The system of equations governing ZGV resonances consists of the guided wave problem and the exceptional mode equation. This represents a singular two-parameter eigenvalue problem that is difficult to solve.

We have presented three very different but complementary computational methods to locate ZGV points in the frequency-wavenumber plane. Their properties are summarized in Table~\ref{tab:comparison_algorithms}. While the direct method is globally convergent but slow and applicable only to small problems, the Newton-type iterative one is very fast but locally convergent. Our third method, which employs the MFRD, combines the direct solution of a regularized problem with the Newton-type procedure. As a result, this method can be applied to rather large problems, does not need initial guesses and is likely to find all ZGV points. We provide the implementation of the three algorithms together with an example in \texttt{GEW ZGV computation}~\cite{gew_zgv_computation}.

While we have applied the concepts to homogeneous but anisotropic plates only, the methods are quite generally applicable. Equations of the same structure are obtained when modeling transversely inhomogeneous waveguides of arbitrary cross-section. Hence, the methods presented here can be used in the same way to find ZGV points of such waveguides.

By fixing $\delta=0$, the system \eqref{eq:3ep_pert} from the method of fixed relative distance (MFRD) 
could alternatively be employed to find the crossing points in the dispersion curves in a general waveguide setting. The computational effort to solve the 
obtained singular three-parameter eigenvalue problem is slightly lower than our direct method to compute ZGV points. Note that such crossing points are known to exist only if the parametrized eigenvalue problem is uniformly decomposable~\cite{Uhlig_SIMAX}.

\begin{acknowledgments}
Bor Plestenjak has been supported by the Slovenian Research Agency (grant N1-0154).

Daniel A. Kiefer and Claire Prada have been supported by LABEX WIFI (Laboratory of Excellence within the French Program “Investments for the Future”) under Reference Nos. ANR-10-LABX-24 and ANR-10-IDEX-0001-02 PSL.
\end{acknowledgments}

\appendix

\section{DERIVATION OF THE WAVEGUIDE PROBLEM} 
\label{sec:derivation_of_the_waveguide_problem}\noindent
Recall the ansatz (\ref{eq:plane_wave_ansatz}) for plane harmonic guided waves:
\begin{equation}
  \ten{\bar{u}}(x,y,t) = \ten{u}(y, k, \omega) \e^{\iu k x - \iu \omega t} \,.
  \label{eq:plane_wave_ansatz_appendix}
\end{equation}
Note that for the above displacements $\partial_z \ten{\bar{u}} = 0$ and $\partial_x \ten{\bar{u}} = \iu k \ten{\bar{u}}$. Acknowledging this, we write the Nabla operator as
\begin{equation}
  \nabla = \dirvec{x} \iu k + \dirvec{y} \partial_y \,.
\end{equation}

Firstly, the proposed motions (\ref{eq:plane_wave_ansatz_appendix}) need to satisfy the equation of motion (\ref{eq:equation_motion_time}). Inserting the above definitions and dropping the complex exponentials yields
\begin{equation}
  (\dirvec{x} \iu k + \dirvec{y} \partial_y) \cdot \ten{c} : (\dirvec{x} \iu k + \dirvec{y} \partial_y) \ten{u} + \rho \omega^2 \ten{u} = \ten{0} \,.
\end{equation}
Note that, per definition, the double contraction consists of two sequential scalar products of the adjacent tensor dimensions, i.e., $\ten{c} : \dirvec{x} \ten{u} = (\ten{c} \cdot \dirvec{x}) \cdot \ten{u}$. Therewith, it is possible to factor out $\ten{u}$ in the above equation. Expanding the products and rearranging terms leads to 
\begin{align}
  [ (\iu k)^2 \dirvec{x} \cdot \ten{c} \cdot \dirvec{x} + 
     \iu k (\dirvec{x} \cdot \ten{c} \cdot \dirvec{y} + \dirvec{y} \cdot \ten{c} \cdot \dirvec{x}) \partial_y &+ \nonumber \\ 
     \dirvec{y} \cdot \ten{c} \cdot \dirvec{y} \partial_y^2 + 
     \omega^2 \rho \ten{I} ] \cdot \ten{u} &= \ten{0} \,,
\end{align}
where $\ten{I}$ is the second order unit tensor. With the definition $\ten{c}_{ij} := \dirvec{i} \cdot \ten{c} \cdot \dirvec{j}$, the above is identical to (\ref{eq:waveguide_problem_motion}). 

Secondly, the motions (\ref{eq:plane_wave_ansatz_appendix}) also need to satisfy the traction-free boundary condition (\ref{eq:BC_free}). Inserting and dropping the complex exponential yields 
\begin{equation}
  \dirvec{y} \cdot \ten{c} : (\dirvec{x} \iu k + \dirvec{y} \partial_y) \ten{u} = \ten{0} \,\hphantom{,} \text{ at } y = \pm h/2 \,.
\end{equation}
By factoring out $\ten{u}$ and multiplying the terms, we immediately obtain the desired result (\ref{eq:waveguide_problem_BC}).

\section{COLLINEARITY OF POWER FLUX AND WAVE VECTOR} 
\label{sec:collinearity_of_power_flux_and_wave_vector}\noindent
The relation between the uncoupling of Lamb- and SH-waves and the collinearity of the power flux and the wave vector are to be discussed. Note first that the SH-polarized displacement component $u_z$ decouples from the Lamb-polarized motions ($u_x$, $u_y$) when the stiffness tensor components satisfy
\begin{equation}
  \label{eq:decoupling_condition}
  c_{z \alpha\beta\gamma} = 0 \ \text{with}\ \alpha, \beta, \gamma \in \{x, y\}\,,
\end{equation}
for details see Ref.~\citen{kiefer_elastodynamic_2022}. Note that the Greek dummy indices run only over $\{x, y\}$. 

The power flux~$\ten{p}$ of the guided waves is obtained from the particle velocity $\ten{v} = -\iu \omega \ten{u}$ and the stress $\ten{T} = \ten{c} : \nabla \ten{u}$. Using again $\nabla = \dirvec{x} \iu k + \dirvec{y} \partial_y$, we obtain
\begin{equation}
  \ten{p} = - \ten{v} \cdot \ten{T} = \iu \omega \ten{u} \cdot \left[ \iu k \ten{c} : \dirvec{x} \ten{u} + \ten{c} : \dirvec{y} \partial_y \ten{u}\right] \,.
\end{equation}

The $y$-component of $\ten{p}$ is always zero due to the flux-free BCs of the guided waves. Hence, the power flux vector~$\ten{p}$ and the wave vector $k \dirvec{x}$ are collinear if the $z$-component of $\ten{p}$ vanishes. Due to the symmetry in $\ten{T}$ (or equivalently in $\ten{c}$), this component is $p_z = \dirvec{z} \cdot (-\ten{v} \cdot \ten{T}) = - \ten{v} \cdot (\dirvec{z} \cdot \ten{T})$, i.e., 
\begin{equation}\label{eq:transverse_power_flux}
  p_z = \iu \omega \iu k \ten{u} \cdot \ten{c}_{zx} \cdot \ten{u} + \iu \omega \ten{u} \cdot \ten{c}_{zy} \cdot \partial_y \ten{u} \,,
\end{equation}
where $\ten{c}_{zx} = \dirvec{z} \cdot \ten{c} \cdot \dirvec{x}$ has components $c_{z i j x}$ with $i,j \in \{x, y, z\}$ and similarly for $\ten{c}_{zy}$.
From the above, we conclude that if $p_z = 0$ for arbitrary $\ten{u}$, this implies the decoupling of Lamb- and SH-polarizations, i.e., (\ref{eq:decoupling_condition}) is satisfied.

On the other hand, let's now assume Lamb wave motions, i.e., (\ref{eq:decoupling_condition}) holds and $u_z = 0$. In this case, the transversal power flux can explicitly be written in terms of the displacement components as
\begin{equation}
  p_z^\mup{Lamb} = \iu \omega \iu k c_{z\alpha\beta x} u_\alpha u_\beta + \iu \omega c_{z\alpha\beta y} u_\alpha \partial_y u_\beta \,,
\end{equation}
where summation is implied over the repeated dummy indices $\alpha, \beta \in \{x, y\}$. From the decoupling condition (\ref{eq:decoupling_condition}) we conclude that $p_z^\mup{Lamb} = 0$. Overall, we can state that the power flux and the wave vector of Lamb waves are collinear if and only if Lamb- and SH-polarizations decouple.

\section{DERIVATION OF THE EXCEPTIONAL MODE EQUATION} 
\label{sec:derivation_of_the_exceptional_mode_equation}\noindent
The PDE of the exceptional mode equation is obtained by inserting the ansatz
\begin{equation}
  \label{eq:appendix_exceptional_mode_ansatz}
  \ten{\bar{u}}(x,y,t) = \left[ \ten{u}'(y, k, \omega) + x \ten{u}(y, k, \omega) \right] \e^{\iu k x - \iu \omega t} 
\end{equation}
into the equation of motion
\begin{align}
  (\dirvec{x} \partial_x + \dirvec{y} \partial_y) \cdot \ten{c} : (\dirvec{x} &\partial_x + \dirvec{y} \partial_y) \ten{\bar{u}}(x,y,t) \nonumber\\ 
  &+ \omega^2 \rho \ten{I} \cdot \ten{\bar{u}}(x,y,t) = \ten{0} \,,
\end{align}
where we have rewritten the Nabla operator as
\begin{equation}
  \label{eq:appendix_nabla}
  \nabla = \dirvec{x} \partial_x + \dirvec{y} \partial_y \,.
\end{equation} 

Using the exceptional mode ansatz (\ref{eq:appendix_exceptional_mode_ansatz}) and dropping the explicit notation of the dependence on $(y, k, \omega)$, one finds
\begin{align}
  \partial_x \ten{\bar{u}}(x,y,t) &= \left[ \ten{u} + \iu k \ten{u}' + \iu k x \ten{u} \right] \e^{\iu k x - \iu \omega t} \,, \\
  \partial_y \ten{\bar{u}}(x,y,t) &= \left[ \partial_y \ten{u}' + x \partial_y \ten{u} \right] \e^{\iu k x - \iu \omega t} \,.
\end{align}
Accordingly, the stress tensor $\ten{c} : \nabla \ten{\bar{u}}(x, y, t)$ is given by
\begin{align}
  \label{eq:appendix_stress}
  \ten{c} : & \left[ \dirvec{x} \partial_x \ten{\bar{u}}(x,y,t) + \dirvec{y} \partial_y \ten{\bar{u}}(x,y,t) \right] = \nonumber \\
  \hphantom{+} & (\ten{c} \cdot \dirvec{x}) \cdot \left[ \ten{u} + \iu k \ten{u}' + \iu k x \ten{u} \right] \e^{\iu k x - \iu \omega t} \nonumber \\
   + &(\ten{c} \cdot \dirvec{y}) \cdot \left[ \partial_y \ten{u}' + x \partial_y \ten{u} \right] \e^{\iu k x - \iu \omega t} \,.
\end{align}
Lastly, after forming the divergence of the above stress field, i.e., contracting from the left with $(\dirvec{x} \partial_x + \dirvec{y} \partial_y)$, we balance with the inertial term. For conciseness, we introduce the notation $\ten{c}_{ij} = \dirvec{i} \cdot \ten{c} \cdot \dirvec{j}$. Performing the multiplications and regrouping the terms, one finally obtains the equation of motion in the form
\begin{align}
   &\hphantom{+} x \underbrace{ \left[ (\iu k)^2 \ten{c}_{xx} + \iu k (\ten{c}_{xy} + \ten{c}_{yx}) \partial_y + \ten{c}_{yy} \partial_y^2 + \omega^2 \rho \ten{I} \right] }_{\ten{W}} \cdot \ten{u} \nonumber \\
   &+ \underbrace{ \left[ (\iu k)^2 \ten{c}_{xx} + \iu k (\ten{c}_{xy} + \ten{c}_{yx}) \partial_y + \ten{c}_{yy} \partial_y^2 + \omega^2 \rho \ten{I} \right] }_{\ten{W}} \cdot \ten{u}' \nonumber \\
   &+ \underbrace{ \left[ 2 \iu k \ten{c}_{xx} + (\ten{c}_{xy} + \ten{c}_{yx}) \partial_y \right] }_{\ten{W}'} \cdot \ten{u} = \ten{0} \,.
\end{align}
As $\ten{u}$ is a guided wave solution, $\ten{W}\cdot \ten{u} = \ten{0}$. Therewith, the above equation is precisely the system (\ref{eq:exceptional_mode_motion}). 

Secondly, the corresponding BC is also needed. The normal vector on the plate's surface is $\pm\dirvec{y}$. The traction free condition is, hence, $\dirvec{y} \cdot \ten{c} : \nabla \ten{u}(x, y, t) = \ten{0}$.  From (\ref{eq:appendix_stress}) we obtain
\begin{equation}
  x \underbrace{ \left[ \iu k \ten{c}_{yx} + \ten{c}_{yy}\partial_y \right] }_{\ten{B}} \cdot \ten{u} + \underbrace{ \left[ \iu k \ten{c}_{yx} + \ten{c}_{yy}\partial_y \right] }_{\ten{B}} \cdot \ten{u}' + \underbrace{ \ten{c}_{yx} }_{\ten{B}'} \cdot \ten{u} = \ten{0} \,,
\end{equation}
where we have again performed the multiplications and resorted the terms. Noting that $\ten{B} \cdot \ten{u} = \ten{0}$, the above is exactly the statement (\ref{eq:exceptional_mode_BC}).

\section{MULTIPARAMETER EIGENVALUE PROBLEM}
\label{sec:mep}\noindent
A \emph{$k$-parameter eigenvalue problem} has the form
\begin{align}
    V_{10}x_1 & = \lambda_1 V_{11}x_1 + \cdots + \lambda_k V_{1k}x_1\nonumber\\
    & \vdots \label{eq:mep}\\
    V_{k0}x_k & = \lambda_1 V_{k1}x_k + \cdots + \lambda_k V_{kk}x_k,\nonumber
\end{align}
where $V_{ij}$ is an $n_i\times n_i$ matrix and $x_i\ne 0$ for $i=1,\ldots,k$.
If \eqref{eq:mep} is satisfied
then a $k$-tuple $(\lambda_1,\ldots,\lambda_k)$ is an \emph{eigenvalue} and 
$x_1\otimes \cdots \otimes x_k$ is the corresponding \emph{eigenvector}. The problem
\eqref{eq:mep}
is related to a system of GEPs
\begin{equation}
    \Delta_1 z =\lambda_1 \Delta_0z,\quad \ldots,\quad \Delta_k z=\lambda_k\Delta_0z,
\end{equation}
where $z=x_1\otimes\cdots\otimes x_k$ and matrices
$$\Delta_0=\left|\begin{matrix}V_{11} & \cdots & V_{1k}\cr
\vdots &  & \vdots \cr 
V_{k1} & \cdots & V_{kk}\end{matrix}\right|_\otimes
$$
and
$$\Delta_i=\left|\begin{matrix}V_{11} & \cdots & V_{1,i-1} & V_{10} & V_{1,i+1} & \cdots & V_{1k}\cr
\vdots &  & \vdots & \vdots & \vdots & & \vdots \cr 
V_{k1} & \cdots & V_{k,i-1} & V_{k0} & V_{k,i+1} & \cdots & V_{kk}\end{matrix}\right|_\otimes
$$
for $i=1,\ldots,k$
are called operator determinants, for details see, e.g., Ref.~\citen{atkinson_multiparameter_1972}. If
$\Delta_0$ is nonsingular then \eqref{eq:mep} is \emph{regular} and matrices 
$\Delta_0^{-1}\Delta_1,\ldots,\Delta_0^{-1}\Delta_k$ commute. A regular 
$k$-parameter EVP \eqref{eq:mep} has $n_1\cdots n_k$ eigenvalues.

\section{EIGENVALUES OF A SINGULAR GEP}
\label{sec:singular}\noindent
Matrices $\Delta_1$ and $\Delta_0$ from \eqref{eq:delta01} form a \emph{singular GEP}, i.e., $\det(\Delta_1-\lambda \Delta_0)\equiv 0$.
Then $\lambda_0\in\mathbb C$ is a finite eigenvalue if $\rank(\Delta_1-\lambda_0\Delta_0)< \nrank(\Delta_1,\Delta_0)$, where
\[
\nrank(\Delta_1,\Delta_0)=\max_{\xi\in\mathbb C}\rank(\Delta_1-\xi \Delta_0)
\]
 is the \emph{normal rank} of the GEP.

\begin{theorem}\label{thm:main}
If $(\iu k,\omega)$ is a solution 
of \eqref{eq:waveguide_problem_discrete} and \eqref{eq:exceptional_mode_discrete_2_block},
then $\lambda=\iu k$ is an eigenvalue of $\Delta_1-\lambda \Delta_0$.
\end{theorem}
\begin{proof}

We introduce $Q(\lambda)=M^{-1}(L_0+\lambda L_1 + \lambda^2 L_2)$. 
Using a block partitioned version of the Kronecker product \cite{TracyJinadasa} we can show that 
$$
\rank(\Delta_1-\lambda \Delta_0)=\rank\left(
\begin{bmatrix}A(\lambda) &  \cr  & B\end{bmatrix}\right),
$$
where $A(\lambda)=$
$$
\begin{bmatrix} 
I\otimes Q(\lambda)-Q(\lambda)\otimes I & \cr
I\otimes Q'(\lambda) & I\otimes Q(\lambda)-Q(\lambda)\otimes I
\end{bmatrix}
$$
and
$$B = \begin{bmatrix}1 & 0\cr 0 & 1\end{bmatrix}\otimes 
(M^{-1}L_2\otimes I-I\otimes M^{-1}L_2).
$$
Clearly, $\rank(\Delta_1-\lambda \Delta_0)=\rank(A(\lambda))+\rank(B)$ and 
$\lambda_0$ is an eigenvalue of $\Delta_1-\lambda \Delta_0$ when
$\rank(A(\lambda_0))<\nrank(A(\lambda))$.

Let us assume that $Q(\lambda)$ is not a uniformly decomposable matrix flow, i.e., there does not exist
a unitary matrix $U$ such that $U^\mup{H}Q(\lambda)U$ has the same block diagonal structure with at least two blocks 
for all $\lambda$. Then, see, e.g., Ref.~\citen{Uhlig_SIMAX}, for a generic $\xi$ all eigenvalues of 
$Q(\xi)$ are distinct. Let $Q(\xi)x_i=\sigma_i x_i$, where $x_i$ is nonzero, for
$i=1,\ldots,n$ and $\sigma_1,\ldots,\sigma_n$ are distinct. Vectors $x_1,\ldots,x_n$ form a basis for $\mathbb C^n$ and vectors
$x_i\otimes x_j$, $i,j=1,\ldots,n$, form a basis for $\mathbb C^{n^2}$.
It is easy to see that $$(I\otimes Q(\xi)-Q(\xi)\otimes I)(x_i\otimes x_j)=(\sigma_j-\sigma_i)(x_i\otimes x_j)$$
for $i,j=1,\ldots,n$. Thus $\rank(I\otimes Q(\xi)-Q(\xi)\otimes I)=n^2-n$ and the null space 
of $I\otimes Q(\xi)-Q(\xi)\otimes I$ is spanned by vectors 
$x_i\otimes x_i$, $i=1,\ldots,n$.

We get 
\begin{equation}
A(\xi)\begin{bmatrix} 0 \cr x_i\otimes x_j \end{bmatrix}=\begin{bmatrix} 0 \cr (\sigma_j-\sigma_i)x_i\otimes x_j \end{bmatrix}\label{eq:thm1}
\end{equation}
and
\begin{equation}
A(\xi)\begin{bmatrix} x_i\otimes x_j \cr 0\end{bmatrix}=\begin{bmatrix}(\sigma_j-\sigma_i)x_i\otimes x_j\cr x_i\otimes Q'(\xi)x_j  \end{bmatrix}\label{eq:thm2}
\end{equation}
for $i,j=1,\ldots,n$.
For $i\ne j$ this gives $2n^2-2n$ linearly independent vectors from the image of $A(\xi)$, while 
vectors $\begin{bmatrix} 0 \cr x_i\otimes x_i \end{bmatrix}$, $i=1,\ldots,n$, are clearly in
the null space of $A(\xi)$. 
What remains are vectors 
\begin{equation}\label{eq:thm_case_ii}
A(\xi)\begin{bmatrix} x_i\otimes x_i \cr 0\end{bmatrix}=\begin{bmatrix}0 \cr x_i\otimes Q'(\xi)x_i  \end{bmatrix}
\end{equation}
for $i=1,\ldots,n$. Each vector $Q'(\xi)x_i$ can be written as a linear combination of $x_1,\ldots,x_n$, i.e.,
\[ 
Q'(\xi)x_i=\sum_{\ell=1}^n \alpha_{i\ell}x_\ell
\]
for $i=1,\ldots,n$. For a generic $\xi$, $\alpha_{ii}\ne 0$ for $i=1,\ldots,n$, so 
vectors from \eqref{eq:thm_case_ii} give additional $n$ linearly independent vectors
from the image of $A(\xi)$ and $\rank(A(\xi))=2n^2-n$.

Let $(\iu k,\omega)$ be a solution of \eqref{eq:waveguide_problem_discrete} and \eqref{eq:exceptional_mode_discrete_2_block}
with the corresponding vectors $u$ and $u'$. If we take $\lambda_0=\iu k$ then 
\begin{align*}
(Q(\lambda_0)+\omega^2I) u&=0,\\
Q'(\lambda_0)u + (Q(\lambda_0)+\omega^2 I)u'&=0.
\end{align*}
It is easy to check that 
\[
A(\lambda_0)\begin{bmatrix}u\otimes u\cr u\otimes u'\end{bmatrix}=0
\]
and we have a vector in the null space that is clearly linearly independent from
vectors that we get from \eqref{eq:thm1}. It follows that 
$\rank(A(\lambda_0))<\nrank(A(\lambda))$ and $\lambda_0$ is an eigenvalue of $\Delta_1-\lambda \Delta_0$.

If $Q(\lambda)$ is a uniformly decomposable flow, the theorem is still valid, the only difference is that
in the proof we have to consider individual blocks in the block diagonal form of $Q(\lambda)$.
\end{proof}

\section{MATERIAL PARAMETERS}
\label{sec:material_parameters}\noindent
The elastic parameters used in the calculation were those of orthotropic austenitic steel provided by Lanceleur~et~al.\cite{lanceleur_use_1993}:
\begin{itemize}
    \item Voigt-notated stiffness in $10^{11}$ \si{\pascal}:
    \begin{equation*}
        C = 
        \begin{bmatrix}
        2.50    &  1.80    &  1.38     &      0   &      0   &      0   \\
        1.80    &  2.50    &  1.12     &      0   &      0   &      0   \\
        1.38    &  1.12    &  2.50     &      0   &      0   &      0   \\
             0  &      0   &     0     &  0.70    &      0   &      0   \\
             0  &      0   &     0     &   0      &   1.17   &      0   \\
             0  &      0   &     0     &   0      &      0   &    0.91 
        \end{bmatrix}
    \end{equation*}
  \item Mass density:
  \begin{equation}
      \rho = \SI{7840}{\kilo\gram\per\meter\cubed}\nonumber
  \end{equation}
\end{itemize}
The stiffness tensor indicated here has been rotated by \SI{90}{\degree} around $\dirvec{x}$ and then around $\dirvec{y}$ (extrinsic, passive rotation), such that the material coordinate system aligns with the one in Fig.~\ref{fig:plate}.

%
%

\end{document}